\documentclass[journal]{IEEEtran}
\usepackage{graphicx}
\usepackage{amsfonts,color}
\usepackage{cite}
\usepackage{epsfig}
\usepackage{latexsym}
\usepackage[T1]{fontenc}
\usepackage{epstopdf}
\usepackage{ amsthm, amssymb}
\usepackage{verbatim}
\usepackage{amsmath}
\usepackage{subfigure}
\usepackage{algorithmic}
\usepackage[algo2e]{algorithm2e}
\usepackage{algorithm}
\usepackage{soul}

\newtheorem{theorem}{Theorem}
\newtheorem{proposition}{Proposition}
\newtheorem{remark}{Remark}
\definecolor{orange}{RGB}{255,107,0}

\makeatletter
\IEEEtriggercmd{\reset@font\normalfont\fontsize{9pt}{9.1pt}\selectfont}
\makeatother
\IEEEtriggeratref{1}

\IEEEoverridecommandlockouts
\usepackage{stackrel}
\usepackage{stackrel}

\usepackage{soul}

\begin{document}
		\pdfoutput=1
    \title{Minimizing low-rank models of high-order tensors: Hardness, span, tight relaxation, and applications} 
    \author{Nicholas~D.~Sidiropoulos,~\IEEEmembership{Fellow,~IEEE,}~Paris A. Karakasis,~\IEEEmembership{Student~Member,~IEEE,}~Aritra Konar,~\IEEEmembership{Member,~IEEE}
		\thanks{Manuscript received May 30, 2023; revised Oct. 19, 2023; accepted Nov. 21, 2023. The Associate Editor coordinating the review was Dr. Yuxin Chen. This research was not supported by any sponsor. The authors believed it is worth pursuing. N. Sidiropoulos and P. Karakasis are with the Dept. of ECE,  University of Virginia, Charlottesville, VA, 22904 USA; e-mail:  (nikos,karakasis)@virginia.edu A. Konar is with the Dept. of EE, KU Leuven, Belgium; e-mail: aritra.konar@kuleuven.be. 

        This paper has supplementary downloadable material available at http://ieeexplore.ieee.org, provided by the author. The material includes software for reproducing the linear parity check decoding experiments. This material is 806 KB in size.

	    }
	    }	
\maketitle
\begin{abstract}
We consider the problem of finding the smallest or largest entry of a tensor of order $N$ that is specified via its rank decomposition. Stated in a different way, we are given $N$ sets of $R$-dimensional vectors and we wish to select one vector from each set such that the sum of the Hadamard product of the selected vectors is minimized or maximized. 
We show that this fundamental tensor problem is NP-hard for any tensor rank higher than one, and polynomial-time solvable in the rank-one case. We also propose a continuous relaxation and prove that it is tight for any rank. For low-enough ranks, the proposed continuous reformulation is amenable to low-complexity gradient-based optimization, and we propose a suite of gradient-based optimization algorithms drawing from projected gradient descent, Frank-Wolfe, or explicit parametrization of the relaxed constraints. We also show that our core results remain valid no matter what kind of polyadic tensor model is used to represent the tensor of interest, including Tucker, HOSVD/MLSVD, tensor train, or tensor ring. Next, we consider the class of problems that can be posed as special instances of the problem of interest. We show that this class includes the partition problem (and thus all NP-complete problems via polynomial-time transformation), integer least squares, integer linear programming, integer quadratic programming, sign retrieval (a special kind of mixed integer programming / restricted version of phase retrieval), and maximum likelihood decoding of parity check codes. We demonstrate promising experimental results on a number of hard problems, including state-of-art performance in decoding low density parity check codes and general parity check codes.      
\end{abstract}
\vspace{-0.25cm}
\section{Introduction}

Finding the smallest or largest entry of a matrix or tensor specified via its low-rank factors is a fundamental problem with numerous applications in science and engineering. The matrix version of the problem has received some attention \cite{Higham16,7373305}, motivated by applications in graph mining (e.g., most significant missing link prediction, nodes that share many neighbors), text/speech/audio similarity search and retrieval (e.g., using text embeddings), and recommender systems (e.g., finding the best item-context combination to recommend to a given user). 

The tensor version of the problem is considerably more powerful, as it allows going beyond bipartite matching / prediction which can be broadly useful in knowledge discovery. As an example, given embeddings of patients, conditions, drugs, clinical trials, therapies, we may be interested in finding the best match, as measured by a function of the pointwise product of the respective embeddings. Tensors can also be used to model high-dimensional probability distributions, wherein low-rank tensor models are used to break the curse of dimensionality while allowing easy marginalization and computation of conditional probabilities, which are crucial for prediction \cite{DBLP:journals/tsp/KargasSF18}. In this context, finding the largest element corresponds to finding the mode / maximum likelihood or maximum {\em a-posteriori} estimate of missing variables, and it is not unusual to encounter very high-order probability tensors indexed by hundreds of categorical input variables. Despite the obvious importance of the tensor version of this problem, there is scant literature about it - we found only \cite{Lu17}, which essentially extends the approach in \cite{7373305} to the low-order tensor case. 

\subsection{Prior work} 

The matrix case was considered in \cite{Higham16}, which proposed using a power-method type algorithm that works directly with the low-rank factors. In \cite{7373305}, the authors developed a randomized ``diamond sampling'' approach for computing the maximum element of the product of two matrices (which could be, e.g., two low-rank factors) in what they called the MAD (Maximum All-pairs Dot-product) search. Their algorithm comes with probabilistic performance guarantees and was demonstrated to work well in practice using a variety of datasets. As already mentioned above, \cite{Lu17} extends the randomized diamond sampling approach in \cite{7373305} to ``star sampling'' for the tensor case. These randomized algorithms do not scale well to high-order tensors, owing to the curse of dimensionality. 

A very different approach for the higher-order tensor version of the problem has been proposed in the computational physics / numerical algebra literature \cite{espigpostprocessing}. The basic idea of \cite{espigpostprocessing} (see also references therein) is as follows. By vectorizing the tensor and putting the resulting (very long) vector on the diagonal of a matrix, the tensor elements become eigenvalues corresponding to coordinate basis eigenvectors. This suggests that the maximum element of the tensor can be computed through a power iteration involving this very large matrix. Of course power iterations implemented naively will have prohibitive complexity (as tensor vectorization produces a very long vector and thus a huge matrix). The idea is therefore to employ a tensor factorization model to ease the matrix-vector multiplication of the diagonal matrix and the interim solution vector. This multiplication can be computed by summing the elements of the pointwise (Hadamard) product of the two vectors -- the vectorized tensor on the matrix diagonal and the interim solution vector. This product can be efficiently computed for various tensor models, but the pointwise multiplication of two tensors of rank $R$ generally has rank up to $R^2$, necessitating conversion back to a rank-$R$ tensor. The natural way to do this is via rank-$R$ least-squares approximation of the higher-rank product. This is a hard and generally ill-posed computational problem, because the best approximation may not even exist. Thus the method cannot be used with a tensor rank decomposition model (known as {\em Canonical Polyadic Decomposition} or {\em CPD}). The pointwise multiplication and least-squares approximation are easier with a so-called {\em tensor train} (TT) model. The Hadamard product of two TT models is another TT model whose wagon (factor) matrices are the Kronecker products of the respective factor matrices of the two TT models. This implies that pointwise multiplication of two TT models has complexity of order $NIR^4$ for an $I^N$ tensor of TT rank $R$ for each wagon. Moreover, rank reduction for TT models is SVD-based and has complexity order $NIR^3$. Hence, reducing the rank of the pointwise multiplication (which is $R^2$ due to the Kronecker product) induces a complexity of order $NI(R^2)^3=NIR^6$, and as a result, only small TT ranks can be supported in practice. Another limitation of \cite{espigpostprocessing} is that it is hard to pin down if and when the iteration will converge, because of the rank reduction step. The numerical experiments in \cite{espigpostprocessing} are interesting but limited in terms of exploring the accuracy-speedup trade-offs that can be achieved.  

Unlike \cite{espigpostprocessing}, our emphasis is on the tensor rank decomposition (CPD) model, in good part because many of the problems that we consider herein admit {\em closed-form} tensor rank decomposition formulations, thus bypassing the need for computational model fitting and/or rank reduction. Another issue is that, as we show for the first time in this paper, the core problem is NP-hard even with a TT model, so there is no intrinsic computational benefit to using one tensor decomposition model over the other. 

\subsection{Contributions}

Our contributions relative to the prior art are as follows: 
\begin{itemize}
    \item We focus on the higher-order tensor version of the problem, and we analyze its computational complexity. We show that the problem is easy when the tensor rank is equal to one, but NP-hard otherwise -- even if the rank is as small as two.
    \item We consider optimization-based (as opposed to algebraic or randomized) algorithms that can be used to compute good approximate solutions in the high-order tensor case. We provide an equivalent ``fluid'' (continuous-variable) reformulation of what is inherently a multi-way selection problem, and a suite of relatively lightweight gradient-based approximation algorithms. The continuous reformulation is derived by showing that a certain continuous relaxation involving a probability distribution instead of hard selection for each mode is actually tight. This is true for any rank, i.e., even for full-rank (unstructured) tensor models. The proposed algorithms take advantage of various optimization frameworks, from projected gradient descent to Franke-Wolfe and various ways of explicitly parametrizing the probability simplex constraints. For low-enough ranks, the associated gradient computations are computationally lightweight, even for very high-order tensors. 
    \item We show that our main results remain valid no matter what kind of polyadic tensor model is used to represent the tensor of interest, including Tucker, HOSVD/MLSVD, tensor train, or tensor ring. 
    \item We explore the ``span'' of the core problem considered, i.e., the class of problems that can be posed as special instances of computing the minimum or maximum element of a tensor from its rank decomposition. We show that this class includes the partition problem (and thus all NP-complete problems via polynomial-time transformation), as well as integer least squares, integer linear programming, integer quadratic programming, sign retrieval (a special kind of mixed integer programming / restricted version of phase retrieval), maximum likelihood decoding of parity check codes, and finding the least inconsistent solution of an overdetermined system of linear equations over a Galois field. 
    \item Finally, we demonstrate promising experimental results on a number of hard problems, including better than state-of-art performance in decoding both low density parity check codes and general parity check codes, and sign retrieval.   
\end{itemize}
     
\vspace{-0.25cm}
\section{Preliminaries}

Consider a matrix ${\bf M}={\bf A}_1 {\bf A}_2^T$, where ${\bf M}$ is $I_1 \times I_2$, ${\bf A}_1$ is $I_1 \times R$ and ${\bf A}_2$ is $I_2 \times R$, with $R \leq \min(I_1,I_2)$ ($R < \min(I_1,I_2)$ if ${\bf M}$ is not full-rank).  Three other ways to write the same relation are ${\bf M} = \sum_{r=1}^R {\bf A}_1(:,r) {\bf A}_2(:,r)^T$, where ${\bf A}(:,r)$ stands for the $r$-th column of ${\bf A}$; ${\bf M} = \sum_{r=1}^R {\bf A}_1(:,r) \circ {\bf A}_2(:,r)$, where $\circ$ stands for the vector outer product; and ${\bf M}(i_1,i_2)=\sum_{r=1}^R {\bf A}_1(i_1,r) {\bf A}_2(i_2,r)$, where ${\bf M}(i_1,i_2)$ denotes an element of matrix ${\bf M}$, with obvious notation. From the latter, notice that element $(i_1,i_2)$ of matrix ${\bf M}$ is the inner product of two row vectors, namely row ${\bf A}_1(i_1,:)$ of matrix ${\bf A}_1$, and row ${\bf A}_2(i_2,:)$ of matrix ${\bf A}_2$. 

We are interested in finding the smallest or largest element of matrix ${\bf M}$ from the factors ${\bf A}_1$ and ${\bf A}_2$. From the latter relation, we see that seeking the smallest (or largest) element of ${\bf M}$ corresponds to seeking a pair of row vectors, one drawn from ${\bf A}_1$ and the other from ${\bf A}_2$, which have the smallest (or largest) possible inner product. One obvious way to do this is to generate all inner products, i.e., the elements of the product ${\bf A}_1 {\bf A}_2^T$, and find the smallest or largest. This entails complexity $I_1 I_2 R$, which can be high, even when $R$ is small -- e.g., consider the case where $I_1$ and $I_2$ are in the order of $10^5$ or more. Is there a way to avoid performing order $I_1 I_2$ computations? When $R=\min(I_1,I_2)$, such complexity is unavoidable, because then the elements of ${\bf M}$ can be freely chosen independently of each other (e.g., from an i.i.d. Gaussian distribution). When $R$ is small however, there seems to be hope. 

Generalizing, consider a tensor ${\bf {\cal T}}$ of order $N$ and size $I_1 \times I_2 \times \cdots \times I_N$, for which we are given a low-rank factorization, namely 
\[
{\bf {\cal T}}=\sum_{r=1}^R {\bf A}_1(:,r) \circ {\bf A}_2(:,r) \circ \cdots \circ {\bf A}_N(:,r),
\]
where ${\bf A}_n$ is $I_n \times R$, and $R$ (when minimal for this decomposition to hold) is the rank of ${\bf {\cal T}}$, in which case the above is known as the {\em canonical polyadic decomposition} (CPD) of ${\bf {\cal T}}$ \cite{7891546}. When the elements of ${\bf {\cal T}}$ are non-negative and the elements of the factor matrices ${\bf A}_n$ are constrained to be non-negative, then the smallest $R$ for which such a decomposition exists is called the non-negative rank of ${\bf {\cal T}}$, which can be higher than the (unconstrained) rank of ${\bf {\cal T}}$. In the sequel, we do not assume that $R$ is minimal; any polyadic decomposition will do for our purposes. In scalar form, we get that
\[
{\bf {\cal T}}(i_1,i_2,\cdots,i_N)=\sum_{r=1}^R {\bf A}_1(i_1,r) {\bf A}_2(i_2,r) \cdots {\bf A}_N(i_N,r),
\]
which reveals that every element of ${\bf {\cal T}}$ is an inner product of $N$ row vectors, one drawn from each of the matrices ${\bf A}_1,\cdots,{\bf A}_N$. An alternative way to write this is using the Hadamard (element-wise) product, denoted by $*$, as
\[
{\bf {\cal T}}(i_1,i_2,\cdots,i_N)= \left( {\bf A}_1(i_1,:) * {\bf A}_2(i_2,:) * \cdots * {\bf A}_N(i_N,:) \right) {\bf 1},  
\]
where ${\bf 1}$ is the $R \times 1$ vector of all $1$'s, used to sum up the $R$ components of the Hadamard product. We see that finding the smallest or largest element of ${\bf {\cal T}}$ is tantamount to finding the smallest or largest inner product of $N$ row vectors, one drawn from each ${\bf A}_n$ matrix. There is an obvious way to do this, by generating all $I_1 \times I_2 \times \cdots \times I_N$ elements of ${\bf {\cal T}}$, at a complexity cost of $R(N-1)$ flops each, but this exhaustive search is much worse than it is in the matrix case. It quickly becomes intractable even for moderate $N$ and $I_N$, e.g., $N=20$ with $I_n=10$, $\forall n \in \left\{1,\cdots,N\right\}$. Is there a more efficient way to do this? 

\section{Theory}

When $R=1$ and all the ${\bf A}_n$ matrices (column vectors ${\bf a}_n$ in this case) only have non-negative entries, it is easy to see that the smallest element of ${\bf {\cal T}}$ is ${\bf {\cal T}}(\check i_1,\cdots,\check i_N)$ with $\check i_n \in \arg \min_{i_n} {\bf a}_n(i_n)$, and likewise the largest is ${\bf {\cal T}}(\hat i_1,\cdots,\hat i_N)$ with $\hat i_n \in \arg \max_{i_n} {\bf a}_n(i_n)$. For $R > 1$ however, the answer is no, in the worst case -- even if $R=2$ and the elements of all the ${\bf A}_n$ matrices are non-negative. We have the following result. 

\begin{theorem}
Finding the smallest element of a tensor from its rank factorization is NP-hard, even when the tensor rank is as small as $R=2$ and the rank-one factors are all non-negative. This means that for large $N$, the worst-case complexity is exponential in $N$. 
\label{theorem:NP}
\end{theorem}

\begin{proof}
We show that an arbitrary instance of the {\em partition problem}, which is known to be NP-complete \cite{GareyJohnson79}, can be converted to a specific instance of the decision version of the problem of interest. This means that if we could efficiently solve any instance of the problem of interest, we would be in position to efficiently solve an arbitrary instance of the partition problem, which is not possible according to the current scientific consensus, unless P=NP. 

Recall the {\em partition problem}, which is as follows. We are given $N$ positive integers $w_1,\cdots,w_N$ (repetitions are allowed), and we wish to know whether there is a subset of indices ${\cal S}$, such that 
\[
\sum_{n \in {\cal S}} w_n = \sum_{n \notin {\cal S}} w_n.  
\]
Consider binary variables $\left\{i_n \in \left\{0,1\right\}\right\}_{n=1}^N$, with $i_n$ designating whether $n \in {\cal S}$ or not, i.e., $i_n=1$ if $n \in {\cal S}$, else $i_n = 0$. Deciding whether a suitable ${\cal S}$ exists is equivalent to deciding whether  
\[
\min_{\left\{i_n \in \left\{0,1\right\}\right\}_{n=1}^N} \left( \sum_{n=1}^N w_n i_n - \sum_{n=1}^N w_n (1-i_n) \right)^2 \stackrel{?}{>} 0.
\]
We will instead use a reformulation of the above which is more conducive for our purposes, namely
\[
\min_{\left\{i_n \in \left\{0,1\right\}\right\}_{n=1}^N} \left( e^{\sum_{n=1}^N w_n i_n} - e^{\sum_{n=1}^N w_n (1-i_n)} \right)^2 \stackrel{?}{>} 0.
\]
Expanding the square and noting that the product term is a constant, we obtain 
\[
\min_{\left\{i_n \in \left\{0,1\right\}\right\}_{n=1}^N} \hspace{-2mm} \left( e^{2\sum_{n=1}^N w_n i_n} + e^{2\sum_{n=1}^N w_n (1-i_n)} \right) \hspace{-0.5mm}\stackrel{?}{>}\hspace{-0.5mm} 2 e^{\sum_{n=1}^N w_n}.
\]
Each of the exponential terms is separable, i.e., a rank-one tensor comprising non-negative factors. For example, the first term is 
\[
e^{\sum_{n=1}^N 2 w_n i_n} = e^{2 w_1 i_1} e^{2 w_2 i_2} \cdots e^{2 w_N i_N}.
\]
It follows that the function to be minimized is a tensor of size $2 \times \cdots \times 2 = 2^N$ and rank $R=2$, with the following $2 \times 2 $ matrix factors
\[
{\bf A}_n = \begin{bmatrix}
1 & e^{2 w_n}\\
e^{2 w_n} & 1
\end{bmatrix},~~~ \forall n \in \left\{1,\cdots,N\right\}.
\]
Hence, the decision version of our problem with $R=2$ and non-negative factors contains every instance of the partition problem as a special case. It follows that the decision version of our problem is NP-hard, and thus its optimization version is NP-hard as well.
\end{proof}

In fact, the same NP-hardness result carries over to other popular tensor models. These include the so-called Tucker model \cite{Tucker1966c},  Multilinear SVD (MLSVD) \cite{MLSVD}, the Tensor Train (TT) decomposition \cite{TensorTrain}, and the Tensor Ring (TR) decomposition \cite{TensorRing}.  
\begin{theorem}
Finding the smallest element of a tensor from its Tucker, MLSVD, TT, or TR factorization is NP-hard. 
\label{theorem:NPall}
\end{theorem}

\begin{proof}
All these models are outer product decompositions that are related. In particular, a rank two CPD is equivalent to a Tucker model with a diagonal $2 \times 2 \times \cdots \times 2$ core matrix. Thus the closed-form rank-two model of the partition problem is also a low multilinear rank ($=2$) Tucker model. We may orthonormalize the loading matrices ${\bf A}_n$ if we want to obtain an MLSVD model that features a dense core, obtained by absorbing the inverse transformations into the core. Thus if we could find the minimum element of every Tucker or MLSVD model of multilinear rank $(2,2,\cdots,2)$ efficiently, we would be in position to solve all instances of the partition problem. 

For the TT decomposition, \cite{CPD-TT} has shown that for low CPD ranks (smaller than or equal to any of the tensor outer size dimensions), it is possible to explicitly construct an equivalent TT model whose lower-order cores exhibit the same low rank as the original high-order CPD model. It follows that we can express our CPD model of the partition problem as a TT model with core ranks equal to two. Finally, the TR model is a generalization of the TT model. In all cases, if we could find the corresponding minimum element efficiently, we would be in position to efficiently solve every instance of the partition problem.  
\end{proof}

\begin{remark}
The decision version of the partition problem is NP-complete, which means that every other NP-complete problem can be transformed in polynomial time to an instance of our problem of interest, i.e., computing the minimum element of a tensor of non-negative rank two. Theorem \ref{theorem:NP} thus speaks not only for the hardness, but also for the span of problems that can be considered within our present framework.    
\end{remark}
\begin{remark}
Theorem \ref{theorem:NP} adds a new problem to a list of tensor problems that are known to be NP-hard \cite{10.1145/2512329}. Notwithstanding, it is interesting that the problem of finding the smallest or largest entry is NP-hard even for tensors of rank as low as $R=2$. 
\end{remark}
\begin{remark}
Multi-way partitioning \cite{MultiWayPartitioning} is a generalization that seeks to split $N$ numbers to $M$ optimally balanced groups. This can also be cast as finding the minimum element of a tensor from its low-rank factors. Let $S=\sum_{n=1}^N w_n$, and $i_n \in \left\{1,\cdots,M \right\}$ be a variable indicating to which of the $M$ groups $w_n$ is assigned to. Consider minimizing the loss function
\[
\sum_{m=1}^M \left( e^{\sum_{n=1}^N w_n 1(i_n=m)}-e^{S/M} \right)^2.
\]
Expanding the square we obtain two rank-one terms and a constant that is irrelevant to optimization. The overall loss is therefore a tensor of rank $R=2M$, and it is easy to write out the associated factor matrices, similar to what we did for the basic partition problem. 
\end{remark}
When the low-rank factors comprise positive and negative values, then even the $R=1$ case can be challenging, as the minimum or the maximum element of each ${\bf a}_n$ can be involved in generating the overall minimum (or maximum), due to the presence of signs. In principle, this seemingly entails explicit or implicit enumeration over $2^N$ possibilities. (If a zero exists in any ${\bf a}_n$, then zero should also be considered as a candidate at the very end.) The good news is that there is structure to this problem: one can invoke the principle of optimality of dynamic programming (DP). The key observation is that \[
\min_{i,j}\boldsymbol{\alpha}_i \boldsymbol{\beta}_j =\min \begin{pmatrix}
\min_i\boldsymbol{\alpha}_i\min_j\boldsymbol{\beta}_j\\
\min_i\boldsymbol{\alpha}_i\max_j\boldsymbol{\beta}_j\\
\max_i\boldsymbol{\alpha}_i\min_j\boldsymbol{\beta}_j\\
\max_i\boldsymbol{\alpha}_i\max_j\boldsymbol{\beta}_j
\end{pmatrix}
\]
and likewise
\[
\max_{i,j}\boldsymbol{\alpha}_i \boldsymbol{\beta}_j =\max\begin{pmatrix}
\min_i\boldsymbol{\alpha}_i\min_j\boldsymbol{\beta}_j\\
\min_i\boldsymbol{\alpha}_i\max_j\boldsymbol{\beta}_j\\
\max_i\boldsymbol{\alpha}_i\min_j\boldsymbol{\beta}_j\\
\max_i\boldsymbol{\alpha}_i\max_j\boldsymbol{\beta}_j
\end{pmatrix}.
\]
Thinking of $\boldsymbol{\alpha}$ as the vector having as elements the minimum and the maximum element products over the first $k$ modes and $\boldsymbol{\beta}$ as ${\bf a}_{k+1}$, we can compute the minimum and maximum up to the $k+1$-th mode using the formulas above. It follows that

\begin{proposition}
\label{prop:one}
Finding the smallest or largest element of a rank-one tensor from $\left\{ {\bf a}_n \in \mathbb{R}^{I_n \times 1} \right\}_{n=1}^N$ can be accomplished via DP at complexity that is linear in $N$. 
\end{proposition}

When $R > 1$, the problem of finding the minimum (or maximum) element of a CPD-factored tensor can be described as follows. One has $N$ buckets of $R$-dimensional vectors, with the $n$-th bucket having $I_n$ vectors -- the rows of ${\bf A}_n$. Finding the minimum element of the CPD-factored tensor is equivalent to selecting a single row vector from each bucket ${\bf A}_n$ such that the inner product of the $N$ resulting vectors is minimized. This is inherently a discrete optimization problem that is NP-hard per Theorem \ref{theorem:NP}. A possible solution strategy is to employ coordinate descent: fixing all indices except $i_n$, we are looking to minimize or maximize over $i_n$ the inner product
\[
{\bf A}_n(i_n,:) {\bf d}_{-n}^T,~\text{ for }~{\bf d}_{-n} := \substack{*\\m \neq n} {\bf A}_m(i_m,:),
\]
which only requires computing the matrix-vector product ${\bf A} {\bf d}_{-n}^T$ and finding its smallest or largest element. Such discrete coordinate descent can be extended to optimizing over small (and possibly randomly chosen) blocks of variables at a time. This is akin to what is known as alternating (block coordinate) optimization in the context of tensor factorization, and it can work quite well in practice for small to moderate $N$. However, such an approach does not seem to scale well with higher $N$, and it is not particularly elegant. We would like to have the option of updating all variables at once, as we do in continuous optimization -- but the problem at hand is discrete and does not appear amenable to tools from continuous optimization. Thankfully, appearances can be deceiving. We have the following result. 

\begin{theorem} Finding the smallest element of a $N$-way tensor in CPD form is {\em equivalent} to the following continuous relaxation involving probability distributions $\left\{ {\bf p}_n\right\}_{n=1}^N$
\begin{equation}
\min_{\left\{ {\bf p}_n \geq 0,~{\bf 1}_{I_n}^T {\bf p}_n=1\right\}_{n=1}^N} \left( \left({\bf p}_1^{T} {\bf A}_1\right) * \cdots *  \left({\bf p}_N^{T} {\bf A}_N\right) \right) {\bf 1}_R. \label{eqn:relaxation}
\end{equation}
\label{theorem:equiv}
\end{theorem}

\begin{proof}
Let $\left\{ \check {\bf p}_n \right\}_{n=1}^N$ denote an optimal solution to (\ref{eqn:relaxation}), and let $\check {\bf q}_{-n} := {\bf A}_n \left( \substack{*\\m \neq n} {\bf A}_m^T \check{\bf p}_m \right)$. Note that the minimum of (\ref{eqn:relaxation}) is then equal to $\check {\bf p}_n^T \check {\bf q}_{-n}$.
The minimum of the inner product $\check {\bf p}_n^T \check {\bf q}_{-n}$ is clearly attained when $\check {\bf p}_n$ is a Kronecker delta that selects a minimum element of $\check {\bf q}_{-n}$ (there might be multiple minimal elements in $\check {\bf q}_{-n}$ that happen to be exactly equal). By virtue of optimality, $\check {\bf p}_n$ cannot combine non-minimal elements of $\check {\bf q}_{-n}$, because that would clearly increase the cost, contradicting optimality. Thus  $\check {\bf p}_n$ generates a convex combination of the possibly multiple equivalent minimal elements of $\check {\bf q}_{-n}$. But because the latter are exactly equal, the same cost is produced by putting all the mass in only one of them. 

Thus, given an optimal solution $\left\{ \check {\bf p}_n \right\}_{n=1}^N$ of (\ref{eqn:relaxation}), we can always round $\check {\bf p}_1$ so that it is integral, without loss of optimality. This yields another optimal solution of (\ref{eqn:relaxation}) to which we can apply the same argument to round $\check {\bf p}_2$ this time, and so on. The proof is therefore complete. 
\end{proof}

\begin{theorem} Finding the smallest element of a $N$-way tensor in Tucker, MLSVD, TT, or TR form is likewise equivalent to the corresponding continuous relaxation. 
\end{theorem}

\begin{proof}
All these decomposition models are fundamentally sums of outer products, i.e., rank-one tensors. Thus they can always be put in the form of multilinear (polyadic) decomposition, albeit such decomposition will not necessarily be {\em canonical}, i.e., of minimal rank, nor will it be unique. Notice however that our proof of the previous Theorem does not assume anything about uniqueness or the number of components in the decomposition. Hence it applies to all these models. As a special case, it applies to full-rank tensors; but then there is no way to avoid the curse of dimensionality, i.e., exponential complexity in gradient computations, see below. Thus it is low-rankness that saves the day. 
\end{proof}

\vspace{-0.45cm}
\section{Methods}

Theorem \ref{theorem:equiv} opens the door to derivative-based optimization. The gradient of the cost function, 
\[f({\bf p}_1, \cdots, {\bf p}_N) = \left( \left({\bf p}_1^{T} {\bf A}_1\right) * \cdots *  \left({\bf p}_N^{T} {\bf A}_N\right) \right) {\bf 1}_R \\
\]
\[
=\left( \substack{*\\n}~ {\bf p}_n^T {\bf A}_n \right) {\bf 1}_R, \] 
with respect to ${\bf p}_n$ is given by 
\[
\nabla_{{\bf p}_n} f = {\bf q}_{-n} := {\bf A}_n \left( \substack{*\\m \neq n} {\bf p}_m^T {\bf A}_m \right)^T.
\]
This is very easy to compute. When ${\bf p}_n^T {\bf A}_n$ contains no zeros, we only have to compute  $\substack{*\\\ m} {\bf p}_m^T  {\bf A}_m$ once, at a cost of $\sum_{n=1}^N I_n R + (N-1)R$ flops, and then divide each time by the leave-one-out factor to compute all the Hadamard products needed. This is followed by a final matrix-vector multiplication for each $n$. The total cost is thus $2 \sum_{n=1}^N I_n R + (2N-1)R$ flops to compute all gradients. To understand what happens when zeros appear, it suffices to consider each column (latent dimension) separately, as multiplications and division happen at the column (rank-one factor) level. If the first element of one and only one of the ${\bf p}_n^T {\bf A}_n$ is zero, then we should also compute the nonzero leave-one-out product of all other first elements and use that only for the gradient update of the ${\bf p}_n$ which generated the said zero element. If the first element of more than one ${\bf p}_n^T {\bf A}_n$ is zero, then all leave-one-out products are zero, hence there is no need to even consider the first dimension in the gradient update, because the corresponding component of the gradient is zero. This implies that for each of the latent dimensions, $r\in\left\{1,\cdots,R\right\}$, we only need to detect if there is a single zero and compute the product of the nonzero elements. Hence, even in the presence of zeros, the worst-case complexity is linear in $N$, as stated above.   

We need to enforce the probability simplex constraints, and projected gradient descent (PGD) a natural choice for this purpose; but we have several other choices. Among them, the Frank-Wolfe algorithm has certain advantages. In our particular context,  computing the minimum inner product of the gradient over the feasible set separates across modes, and for each mode it boils down to finding the minimum element of the corresponding part of the gradient vector. Thus Frank-Wolfe bypasses projection onto the simplex, which requires an iterative algorithm. Frank-Wolfe applied to nonconvex cost functions with convex constraints enjoys nice convergence guarantees \cite{DBLP:journals/corr/Lacoste-Julien16} -- note that our multilinear cost function is nonconvex, but the probability simplex constraints are convex. 

Another way to bypass the projection onto the simplex is to introduce what is sometimes referred to as the {\em Hadamard parametrization} of a probability distribution ${\bf p}(i)=({\bf s}(i))^2$, where ${\bf p}$ is the Hadamard product of ${\bf s}$ with itself \cite{HadamardParametrization}; or the {\em amplitude parametrization} ${\bf p}(i)=|{\bf s}(i)|^2$ of the quantum literature, where ${\bf s}$ is complex and ${\bf p}$ is the Hadamard product of ${\bf s}$ and the conjugate of ${\bf s}$. With these parametrizations ${\bf p}$ always has non-negative elements and sums up to one if and only if $||{\bf s}||_2=1$. Thus the simplex constraint is transformed to a unit sphere constraint. One advantage of this is that projection of any vector on the unit sphere only involves normalization, i.e., ${\bf s}=\frac{\bf s}{||{\bf s}||_2}$. A drawback is that the unit sphere is not a convex set, which complicates convergence analysis. The third option when it comes to parametrizing the simplex constraint is to use ${\bf p}(i)=\frac{e^{{\bf v}(i)}}{\sum_{j=1}^I e^{{\bf v}(j)}}$, where vector ${\bf v}$ is unconstrained. This parametrization emerges from mirror descent using negative entropy, but it can also be motivated from the viewpoint of turning simplex-constrained optimization into an unconstrained problem on which gradient descent can be applied. Following the latter viewpoint and applying the chain rule
\[
\frac{\partial f}{\partial {\bf v}_n(i)} = {\bf p}_n(i) {\bf g}_n(i) - {\bf p}_n(i) {\bf p}_n^T {\bf g}_n,
\]
where ${\bf g}_n := \nabla_{{\bf p}_n} f$. 

The aforementioned simplex parametrization approaches can model any finite distribution. In our particular context, we know that the sought distribution can be restricted to be unimodal -- after all, the final solution can be ``rounded'' to a Kronecker delta without loss of optimality per the proof of Theorem \ref{theorem:equiv}. Towards this end, we can use as potential map ${\bf v}(i)$ a discrete analog of the exponent of a Gaussian, namely 
\[
{\bf v}(i)=-\left(\frac{i-\mu}{\sigma}\right)^2,
\]
where $\mu \in \mathbb{R}$ is a location parameter and $\sigma \in \mathbb{R}$ controls the spread of the distribution. Using the chain rule again, the derivatives for updating these two ``root'' parameters are
\[
\frac{\partial f}{\partial \mu}= \sum_{i=1}^I \frac{\partial f}{\partial {\bf v}_n(i)} \frac{2(i-\mu)}{\sigma^2}
\]
and
\[
\frac{\partial f}{\partial \sigma}= \sum_{i=1}^I \frac{\partial f}{\partial {\bf v}_n(i)} \frac{2(i-\mu)^2}{\sigma^3}. 
\]

Frank-Wolfe, PGD, and the exponential parametrization work well for random instances of the partition problem, subject to suitable tuning of the step-size related parameters. Pseudo-code listings of these four algorithms are provided as Algorithms \ref{algo:FW}, \ref{algo:PSPGD}, \ref{algo:EXP}, \ref{algo:DGP} 
for the Frank-Wolfe, PGD, exponential, and ``discrete Gaussian'' parametrization, respectively. For the partition problem, the Frank-Wolfe Algorithm \ref{algo:FW} appears to offer the best performance and the lowest complexity in our proof-of-concept experiments. Each of these algorithms is useful in different application contexts, as we will see. Algorithm \ref{algo:DGP} uses the most compact parametrization of all, and is often the most efficient in terms of iterations needed for convergence.  

For the Frank-Wolfe method, for which there is proof that the algorithm attains a stationary point at a rate of ${\cal O}(\frac{1}{\sqrt t})$, where $t$ is the number of iterations, computing the associated curvature constant is non-trivial. In the matrix case ($N=2$), the Lipschitz constant of the gradient of the cost function for our problem, which can be used to bound the curvature constant in Frank-Wolfe \cite{DBLP:journals/corr/Lacoste-Julien16} is determined by the principal singular value of ${\bf A}_1 {\bf A}_2^T$. To see this, note that for $N=2$ the cost function can be written as
\[
f({\bf p}_1,{\bf p}_2)= {\bf p}_1^T {\bf A}_1 {\bf A}_2^T {\bf p}_2,
\]
the gradient of which is given by
\[
\nabla_{\left[ {\bf p}_1^T {\bf p}_2^T \right]^T} f = \begin{bmatrix}
{\bf A}_1 {\bf A}_2^T {\bf p}_2\\
{\bf A}_2 {\bf A}_1^T {\bf p}_1
\end{bmatrix}.
\]
It follows that
\[
\left|\left| \nabla_{\left[ {\bf p}_1^T {\bf p}_2^T \right]^T} f - \nabla_{\left[ {\bf q}_1^T {\bf q}_2^T \right]^T} f\right|\right|_2^2 = \left|\left| \begin{bmatrix}
{\bf A}_1 {\bf A}_2^T ({\bf p}_2 - {\bf q}_2) \\
{\bf A}_2 {\bf A}_1^T ({\bf p}_1 - {\bf q}_1) 
\end{bmatrix} \right|\right|_2^2,
\]
\[
\leq \left(\sigma_{\text{max}}\left( {\bf A}_1 {\bf A}_2^T \right)\right)^2 \left|\left| \begin{bmatrix}
{\bf p}_2 - {\bf q}_2 \\
{\bf p}_1 - {\bf q}_1 
\end{bmatrix} \right|\right|_2^2,
\]
hence
\[
\left|\left| \nabla_{\left[ {\bf p}_1^T {\bf p}_2^T \right]^T} f - \nabla_{\left[ {\bf q}_1^T {\bf q}_2^T \right]^T} f\right|\right|_2 \hspace{-2mm}\leq \sigma_{\text{max}}\left( {\bf A}_1 {\bf A}_2^T \right) \left|\left| \begin{bmatrix}
{\bf p}_2 - {\bf q}_2 \\
{\bf p}_1 - {\bf q}_1 
\end{bmatrix} \right|\right|_2.
\]
This is unfortunate, for computing $\sigma_{\text{max}}\left( {\bf A}_1 {\bf A}_2^T \right)$ is more costly than finding the smallest element of the product ${\bf A}_1 {\bf A}_2^T$ by brute-force. In the $N$-way case, the situation is more complicated. Note that for $N=2$ the difference of the gradients is linear in the difference of the mode distributions. This is not true when $N>2$.

\subsection{Complexity}

Assuming $I_n=I$, $\forall n$ for brevity, the cost of computing all gradients with respect to the mode distributions ${\bf p}_n$ is of order $NIR$, as we have seen. The Frank-Wolfe algorithm maintains this complexity order per iteration. When employing the exponential parametrization of the mode distributions or the ``discrete Gaussian'' parametrization, the extra gradient back-propagation steps have complexity $NI$, thus again maintaining overall per-iteration complexity of  order $NIR$. The actual complexity of the overall algorithm depends on a number of critical parameters, including the choice of gradient stepsize, the maximum number of iterations allowed (``hard-stop''), and the initialization used -- intelligent / application-specific or random. We are using these algorithms to tackle NP-hard problems, so initialization does matter. In certain applications, such as parity check decoding, there is a natural initialization that we can use (the channel output bits), but in others, like the sign retrieval application that we will consider in some detail, there is no ``natural'' initialization that we can use. Through experimentation, we have found that one initialization that works well in many cases is to use the DP-based algorithm to compute an optimal solution for each rank-one factor separately, and then pick among those the one that is best for the higher-rank tensor minimization problem. This often gives a good ``universal'' (application-agnostic) initialization the complexity of which is of order $NIR+NR^2$. 

Tuning the gradient stepsize is not difficult in our experience, but it is application-dependent. The maximum number of gradient iterations is set between $300$ and $3,000$ in all our experiments, even for high-order problems (high $N$). Thus complexity is always polynomial of order $NIR$, but there is no guarantee that the optimal solution will be found. Notwithstanding, as we show in our experiments, the optimization performance attained is often state-of-art.  


\begin{algorithm}[h]
	\caption{Min CPD via Frank-Wolfe w/ adaptive stepsize}
	\label{algo:FW}
 \begin{algorithmic}
\STATE\hspace{-1em}\textbf{Input:} $\left\{ {\bf A}_n \right\}_{n=1}^N$, \text{curvature parameter $C$}.
 
\STATE\hspace{0.5em}1. \small{Initialize mode distributions ${\bf p}_n$ randomly,} $\forall n \in \left\{1,\cdots,N\right\}$.

\STATE\hspace{0.5em}2. \small{repeat}

\STATE\hspace{0.5em}3. \small{Compute mode gradients, update directions, adaptive stepsize:}

\STATE\hspace{1.5em} Set $g_t=0$

\STATE\hspace{1.5em} \small{for n=1 to N}

\STATE\hspace{2.5em} \small{Compute $\nabla_{{\bf p}_n} f = {\bf A}_n \left( \substack{*\\m \neq n} {\bf p}_m^T {\bf A}_m \right)^T$}

\STATE\hspace{2.5em} \small{Find $v_n^*=\min_{i_n=1}^{I_n} \nabla_{{\bf p}_n} f(i_n)$}  

\STATE\hspace{2.5em} \small{Set ${\cal M}_n=\left\{i_n ~|~ \nabla_{{\bf p}_n} f(i_n)=v_n^*\right\}$.}

\STATE\hspace{2.5em} \small{Set ${\bf d}_n(i_n)=1/|{\cal M}_n|, \forall i_n \in {\cal M}_n$}

\STATE\hspace{2.5em} \small{Accumulate $g_t = g_t + ({\bf p}_n - {\bf d}_n)^T \nabla_{{\bf p}_n} f$}

\STATE\hspace{1.5em} \small{end for}

\STATE\hspace{1.5em} \small{Set $\lambda_t = \min(\frac{g_t}{C},1)$}

\STATE\hspace{0.5em}4. \small{Update mode distributions:}

\STATE\hspace{1.5em} \small{for n=1 to N}

\STATE\hspace{2.5em} \small{${\bf p}_n = (1-\lambda_t) {\bf p}_n + \lambda_t {\bf d}_n$}

\STATE\hspace{1.5em} \small{end for}

\STATE\hspace{0.5em}5. \small{until convergence criterion met} 
 \STATE\hspace{-1em} \textbf{Output:} $i_n^* \in \arg \max_{i \in I_n} {\bf p}_n(i)$.      
\end{algorithmic}
\end{algorithm}

\vspace*{.25in}

\begin{algorithm}[h]
	\caption{Min CPD via simplex PGD with momentum}
	\label{algo:PSPGD}
\begin{algorithmic}
\STATE	\hspace{-1em}\textbf{Input:} $\left\{ {\bf A}_n \right\}_{n=1}^N$, \text{step size, momentum parameters} $\lambda, \beta$, resp.

\STATE\hspace{1em}1. \small{Init. mode distributions ${\bf p}_n$ randomly, $\forall n \in \left\{1,\cdots,N\right\}$.}

\STATE\hspace{1em}2. \small{Initialize gradients ${\bf g}_n={\bf 0}_{I_n \times 1}$, $\forall n \in \left\{1,\cdots,N\right\}$.}

\STATE\hspace{1em}3. \small{repeat}

\STATE\hspace{1em}4. \small{Compute mode gradients, update mode distributions:}

\STATE\hspace{2em} \small{for n=1 to N}

\STATE\hspace{3em} \small{Compute $\nabla_{{\bf p}_n} f = {\bf A}_n \left( \substack{*\\m \neq n} {\bf p}_m^T {\bf A}_m \right)^T$}

\STATE\hspace{3em} \small{Accumulate momentum ${\bf g}_n = (1-\beta) {\bf g}_n + \beta \nabla_{{\bf p}_n} f$}

\STATE\hspace{3em} \small{Update and project onto simplex ${\bf p}_n = {\cal P}_{\Omega} \left({\bf p}_n - \lambda {\bf g}_n\right)$}

\STATE\hspace{2em} \small{end for}

\STATE\hspace{1em}5. \small{until convergence criterion met} 

 \STATE 	\hspace{-1em}\textbf{Output:} $i_n^* \in \arg \max_{i \in I_n} {\bf p}_n(i)$.      
\end{algorithmic}
\end{algorithm}

\vspace*{.25in}

\begin{algorithm}[h]
	\caption{Min CPD via Exponential parametrization}
	\label{algo:EXP}
\begin{algorithmic}
\STATE	\hspace{-1em}\textbf{Input:} $\left\{ {\bf A}_n \right\}_{n=1}^N$, \text{step size parameter} $\lambda$.

\STATE	\hspace{1em}1. \small{Init. ${\bf v}_n$ randomly from i.i.d. normal distribution.}

\STATE\hspace{1em}2. \small{repeat}

\STATE\hspace{1em}3. \small{Compute mode gradients, update mode distributions:}

\STATE\hspace{2em} \small{for n=1 to N}

\STATE\hspace{3em} \small{Compute ${\bf g}_n := \nabla_{{\bf p}_n} f = {\bf A}_n \left( \substack{*\\m \neq n} {\bf p}_m^T {\bf A}_m \right)^T$}

\STATE\hspace{3em} \small{Compute $\frac{\partial f}{\partial {\bf v}_n(i)} = {\bf p}_n(i) {\bf g}_n(i) - {\bf p}_n(i) {\bf p}_n^T {\bf g}_n, \forall i$}

\STATE\hspace{3em} \small{Update ${\bf v}_n(i) = {\bf v}_n(i) - \lambda \frac{\partial f}{\partial {\bf v}_n(i)}, \forall i$}

\STATE\hspace{3em} \small{Update ${\bf p}_n(i)=\frac{e^{{\bf v}_n(i)}}{\sum_{j=1}^I e^{{\bf v}_n(j)}}, \forall i$}

\STATE\hspace{2em} \small{end for}

\STATE\hspace{1em}4. \small{until convergence criterion met} 

\STATE 	\textbf{Output:} $i_n^* \in \arg \max_{i \in I_n} {\bf p}_n(i)$.      
\end{algorithmic}
\end{algorithm}

\newpage
\begin{algorithm}[h]
	\caption{Min CPD via ``discrete Gaussian'' parametrization}
	\label{algo:DGP}
\begin{algorithmic}
\STATE	\hspace{-1em}\textbf{Input:} $\left\{ {\bf A}_n \right\}_{n=1}^N$, \text{step size parameter} $\lambda$.

\STATE	\hspace{1em}1. \small{Init. ${\bf v}_n$ randomly from i.i.d. normal distribution.}

\STATE\hspace{1em}2. \small{repeat}

\STATE\hspace{1em}3. \small{Compute mode gradients, update mode distributions:}

\STATE\hspace{2em} \small{for n=1 to N}

\STATE\hspace{3em} \small{Compute ${\bf g}_n := \nabla_{{\bf p}_n} f = {\bf A}_n \left( \substack{*\\m \neq n} {\bf p}_m^T {\bf A}_m \right)^T$}

\STATE\hspace{3em} \small{Compute $\frac{\partial f}{\partial {\bf v}_n(i)} = {\bf p}_n(i) {\bf g}_n(i) - {\bf p}_n(i) {\bf p}_n^T {\bf g}_n, \forall i$}

\STATE\hspace{3em} \small{Compute $\frac{\partial f}{\partial \mu_n}= \sum_{i=1}^I \frac{\partial f}{\partial {\bf v}_n(i)} \frac{2(i-\mu_n)}{\sigma_n^2}$}

\STATE\hspace{3em} \small{Compute $\frac{\partial f}{\partial \sigma_n}= \sum_{i=1}^I \frac{\partial f}{\partial {\bf v}_n(i)} \frac{2(i-\mu_n)^2}{\sigma_n^3}$}

\STATE\hspace{3em} \small{Update $\mu_n = \mu_n - \lambda \frac{\partial f}{\partial \mu_n}$}

\STATE\hspace{3em} \small{Update $\sigma_n = \sigma_n - \lambda \frac{\partial f}{\partial \sigma_n}$}

\STATE\hspace{3em} \small{Update ${\bf v}_n(i) = -\left(\frac{i-\mu_n}{\sigma_n}\right)^2, \forall i$}

\STATE\hspace{3em} \small{Update ${\bf p}_n(i)=\frac{e^{{\bf v}_n(i)}}{\sum_{j=1}^I e^{{\bf v}_n(j)}}, \forall i$}

\STATE\hspace{2em} \small{end for}

\STATE\hspace{1em}4. \small{until convergence criterion met} 

\STATE 	\textbf{Output:} $i_n^* \in \arg \max_{i \in I_n} {\bf p}_n(i)$.      
\end{algorithmic}
\end{algorithm}

\section{How expressive is the class of problems considered?}
We have seen that the class of problems that can be viewed as special instances of finding the minimum element of a tensor from its rank-one factors is broad -- it contains the partition problem, and thus any NP-complete problem can be transformed in polynomial time to our problem of interest. Still, such transformation may not be obvious, and one wonders whether there exist broadly useful classes of NP-hard problems that are directly amenable, or easily transformable, to instances of the problem of interest. As we will see next, the answer is affirmative for various important optimization models that are frequently used in engineering.  

\subsection{Integer Linear Programming} 

Consider the Integer Linear Programming (ILP) problem
\begin{equation}
    \begin{split}
    \min_{\mathbf{x}\in\mathbb{S}^N}& ~\mathbf{c}^T\mathbf{x}\\
    \text{s.t.}~&~ \mathbf{H x}\leq \mathbf{b},
    \end{split}
\end{equation}
where $\mathbf{H}\in\mathbb{R}^{M\times N}$, $\mathbf{b}\in\mathbb{R}^{M}$, $\mathbf{c}\in\mathbb{R}^{N}$, and $\mathbb{S}^N$ is a finite subset of $\mathbb{R}^N$. $\mathbb{S}^N$ will typically be a finite lattice, i.e., the Cartesian product of finite subsets of $\mathbb{R}$. In the supplementary material, we show that it is possible to transform this ILP problem to minimizing  
\[
      \min_{\mathbf{x}\in\mathbb{S}^N} \left\{ e^{t \mathbf{c}^T \mathbf{x}} + 
 ~\sum_{m=1}^{M} \lambda_m e^{t(\mathbf{c}+\rho{\mathbf{h}}_m)^T\mathbf{x}} \right\},
\]
where $\lambda_m := e^{- t \rho \mathbf{b}(m)}$, for sufficiently large (problem-specific) $\rho$ and $t$. Every exponential inside the brackets is a rank-one tensor (a separable function of the variables in ${\bf x}$), and thus the above is a tensor of order $N$ and rank at most $M+1$, which is very low compared to the maximal possible rank ($I^{N-1}$ when $\mathbb{S}^N = {\cal I}^N$). Here ${\cal I}^N$ is the Cartesian product of $N$ copies of ${\cal I}$ (note that $\mathbb{S}^N$ need not be a Cartesian product in general -- we are using slightly overloaded notation).

\subsection{Integer Least Squares}

Consider the integer least squares (ILS) problem 
\begin{equation}
    \min_{\mathbf{x}\in\mathbb{S}^N} \left|\left|\mathbf{H} \mathbf{x} -  \mathbf{b} \right|\right|_2^2,
\end{equation}
where $\mathbf{H}\in\mathbb{R}^{M\times N}$, $\mathbf{b}\in\mathbb{R}^{M}$. Let $\mathbf{G} := \mathbf{H}^T \mathbf{H}$ and $\mathbf{c}:= 2 \mathbf{H}^T \mathbf{b}$. Then, it is easy to see that 
\begin{equation}
    \begin{split}
&~~~~~~~~~~\min_{\mathbf{x}\in\mathbb{S}^N} \left|\left|\mathbf{H} \mathbf{x} -  \mathbf{b} \right|\right|_2^2  \equiv \min_{\mathbf{x}\in\mathbb{S}^N} \mathbf{x}^T \mathbf{G} \mathbf{x} - \mathbf{c}^T \mathbf{x} \\
&\phantom{xx}\equiv \min_{\mathbf{x}\in\mathbb{S}^N} \sum_{n=1}^N \sum_{m=1}^N \mathbf{G}(n,m) \mathbf{x}(n) \mathbf{x}(m) -  \sum_{n=1}^N \mathbf{c}(n) \mathbf{x}(n).  
\raisetag{37pt}
\label{eqn:ILS}
    \end{split}
\end{equation}
Note that a term of the form\footnote{We may assume without loss of generality that $m \geq n$. Here, $\mathbf{x}^0$ is element-wise exponentiation, i.e., a vector of all ones.}
\[
\gamma \mathbf{x}(n) \mathbf{x}(m) = \gamma \mathbf{x}^0(1) \cdots \mathbf{x}^0(n-1) \mathbf{x}(n) \mathbf{x}^0(n+1) \cdots 
\]
\[
\mathbf{x}^0(m-1) \mathbf{x}(m) \mathbf{x}^0(m+1) \cdots \mathbf{x}^0(N)
\]
is separable (rank-one), and so is $\gamma (\mathbf{x}(n))^2$.
It follows that the cost function in \eqref{eqn:ILS} has rank at most $\frac{N(N-1)}{2}+N$, which is again very low compared to the maximal possible rank ($I^{N-1}$ when $\mathbb{S}^N = {\cal I}^N$). Note that here we have used the symmetry of $\mathbf{G} := \mathbf{H}^T \mathbf{H}$ to reduce the required number of rank-one terms. 

\subsection{Integer Quadratic Programming} 

Consider the integer indefinite quadratic problem $\min_{\mathbf{x}\in\mathbb{S}^N} \mathbf{x}^T \mathbf{Q} \mathbf{x}$, where $\mathbf{Q}$ is not necessarily positive semidefinite, or even symmetric. Note that we can always symmetrize without loss of generality, as $\mathbf{x}^T \mathbf{Q} \mathbf{x} = (\mathbf{x}^T \mathbf{Q} \mathbf{x})^T = \mathbf{x}^T \mathbf{Q}^T \mathbf{x}$, and thus $\mathbf{x}^T \mathbf{Q} \mathbf{x} = \mathbf{x}^T \frac{(\mathbf{Q} + \mathbf{Q}^T)}{2} \mathbf{x} = \mathbf{x}^T \mathbf{G} \mathbf{x}$, with $\mathbf{G} := \frac{(\mathbf{Q} + \mathbf{Q}^T)}{2}$. It follows that integer quadratic programming corresponds to finding a minimum element of a CPD model of rank $\frac{N(N-1)}{2}$. 

\subsection{Mixed integer programming: Sign retrieval} So far we have considered optimization problems with purely categorical (``integer'') variables. As an example of a mixed integer problem that falls under our framework, we next consider {\em sign retrieval} (e.g., see \cite{LESHEM2018463,suzuki2022compressing}). This is a special case of {\em phase retrieval} \cite{YPR-7078985,YPR-7559964}, and both have important applications in a broad range of disciplines, from optical imaging \cite{YPR-7078985} to wireless communication  -- where it can be used for channel estimation from coarse channel quality measurements \cite{Qiu-8253866}. The starting point of the sign retrieval problem is the measurement model
\[
{\bf y} = \left| {\bf A} {\bf x} + {\bf v} \right|,
\]
where ${\bf x} \in \mathbb{R}^{N \times 1}$ is a vector of unknowns to be estimated, ${\bf A} \in \mathbb{R}^{M \times N}$ with $M > N$ (typically $M >> N$) is known, ${\bf v}$ is additive white Gaussian noise, $|\cdot|$ takes the absolute value of its argument, and ${\bf y} \in \mathbb{R}_+^{M \times 1}$ is the sign-less vector of measurements. Treating ${\bf s} = \text{sign}({\bf y})$ and ${\bf x}$ as deterministic unknowns, maximum likelihood estimation amounts to 
\[
\min_{{\bf s} \in {\color{black} \{\pm 1\}}^{M \times 1}, ~ {\bf x} \in \mathbb{R}^{N \times 1}} ||{\bf D}({\bf s}) {\bf y} - {\bf A} {\bf x}||_2^2, 
\]
where ${\bf D}({\bf s})$ is a diagonal matrix holding the elements of ${\bf s}$ on its diagonal. The problem is separable with respect to the continuous parameters ${\bf x}$. That is, solving for ${\bf x}$ as a function of ${\bf s}$, ${\bf x} = ({\bf A}^T {\bf A})^{-1} {\bf A}^T {\bf D}({\bf y}) {\bf s}$ and substituting the result back into the cost function, we obtain
\[
\min_{{\bf s} \in {\color{black}\{\pm 1\}}^{M \times 1}} ||({\bf I} - {\bf A} ({\bf A}^T {\bf A})^{-1} {\bf A}^T) {\bf D}({\bf y}) {\bf s})||_2^2,  
\]
where we have used ${\bf D}({\bf s}) {\bf y} = {\bf D}({\bf y}) {\bf s}$. Expanding and using the idempotence of $({\bf I} - {\bf A} ({\bf A}^T {\bf A})^{-1} {\bf A}^T)$, we can further rewrite the problem as
\[
\min_{{\bf s} \in {\color{black} \{\pm 1\}}^{M \times 1}} {\bf s}^T {\bf Q} {\bf s},~~~ \text{with}~{\bf Q}:={\bf D}({\bf y}) ({\bf I} - {\bf A} ({\bf A}^T {\bf A})^{-1} {\bf A}^T) {\bf D}({\bf y}).
\]
It follows from the preceding subsection on integer quadratic programming that the sign retrieval problem corresponds to finding a minimum element of a CPD model of rank $\frac{M(M-1)}{2}$.  

\subsection{Maximum likelihood / minimum distance decoding of parity check codes over Galois Fields}
Consider a parity check code \cite{LinCostello} over GF(2) with $M \times N$ parity check matrix ${\bf C}$, where $M$ is the number of parity checks, $N$ is the codeword length, and the code rate is $\frac{K}{N}$, where $K:=N-M$. A certain $N \times 1$ binary vector ${\bf x} \in \left\{0,1\right\}^N$ is a valid codeword if and only if ${\bf C} {\bf x}={\bf 0}_{M \times 1}$ in GF(2) modulo 2 arithmetic, i.e., $\text{mod}({\bf C} {\bf x},2)={\bf 0}_{M \times 1}$ in real arithmetic. Assuming that the coded bits are transmitted over a memoryless binary symmetric channel (BSC) with cross-over probability $p<0.5$, or over an additive white Gaussian noise (AWGN) channel, maximum likelihood decoding reduces to minimum Hamming or Euclidean distance decoding, respectively. Noting that for $\left\{0,1\right\}$-encoding of the channel input ${\bf x}$ and output ${\bf y}$, Euclidean distance squared is equal to Hamming distance, we can write both using real arithmetic as   
\begin{equation}
\min_{{\bf x} \in \left\{0,1\right\}^N | \text{mod}({\bf C} {\bf x},2)={\bf 0}} ||{\bf y}-{\bf x}||_2^2.
\label{eq:MLPCC}
\end{equation}
Optimal decoding is an NP-hard problem for most ``irregular'' codes of current interest\footnote{E.g., Convolutional codes with short memory can be optimally and efficiently decoded using DP.}, including low-density parity check (LDPC) codes which can come close to attaining the Shannon limit. These are decoded using an iterative message passing technique known as belief propagation, which performs well for code graphs that are free of short loops \cite{IntroLDPC,Eleftheriou}. We will next show that the same optimal decoding problem can be approached in a very different way: as the problem of computing the minimum element of a low-rank tensor. Consider the following problem
\begin{equation}
\label{eq:MLPCCasminCPD}
\min_{{\bf x} \in \left\{0,1\right\}^N}\hspace{-1mm} - \sum_{m=1}^M (-1)^{{\bf C}(m,:) {\bf x}} + \frac{1}{e} \left( 1 + \frac{1}{N}\right)^{\sum_{n=1}^N ({\bf y}(n)-{\bf x}(n))^2}\hspace{-10mm}.
\end{equation}
Note that $(-1)^{{\bf C}(m,:) {\bf x}}=(-1)^{\sum_{n=1}^N {\bf C}(m,n) {\bf x}(n)}=\prod_{n=1}^N (-1)^{{\bf C}(m,n) {\bf x}(n)}$, which is a rank-one tensor. Likewise, $\left( 1 + \frac{1}{N}\right)^{\sum_{n=1}^N ({\bf y}(n)-{\bf x}(n))^2}=\prod_{n=1}^N \left( 1 + \frac{1}{N}\right)^{({\bf y}(n)-{\bf x}(n))^2}$ is separable, and thus a rank-one tensor. The overall cost function in (\ref{eq:MLPCCasminCPD}) is therefore a tensor of rank (at most) $M+1$. We have the following result.\\ 
\begin{proposition}
Solving (\ref{eq:MLPCCasminCPD}) is equivalent to solving (\ref{eq:MLPCC}). 
\end{proposition}
\begin{proof}
Note that each term of the sum on the left is equal to $1$ when the corresponding parity equation is satisfied, or $-1$ otherwise. With the minus sign up front of the sum, minimization of this term will produce a valid codeword that satisfies all parity checks and yield an overall value $-M$. Any constraint that is violated will increase this term by $2$. The term on the right is monotonically increasing with the squared loss  $||{\bf y}-{\bf x}||_2^2=\sum_{n=1}^N ({\bf y}(n)-{\bf x}(n))^2$. The maximum value of the latter is $N$, while $\left( 1 + \frac{1}{N}\right)^N$ is upper bounded by its limit as $N \rightarrow  \infty$, which is the base of the natural logarithm, $e$. It follows that the term on the right is strictly less than $1$ (and lower bounded by $\frac{1}{e}$). Since the all-zero codeword satisfies all parity checks, it follows that the optimum solution to (\ref{eq:MLPCCasminCPD}) should have cost less than or equal to $-M+1$. Moreover, this value cannot be attained by any ${\bf x}$ which does not satisfy all parity checks, because such a ${\bf x}$ would incur a cost of at least $-M+2$, even discarding the second term in (\ref{eq:MLPCCasminCPD}). Hence, any solution of (\ref{eq:MLPCCasminCPD}) must satisfy all parity checks. Among all ${\bf x}$ that satisfy all parity checks, those that minimize $\sum_{n=1}^N ({\bf y}(n)-{\bf x}(n))^2$ yield the least overall cost in (\ref{eq:MLPCCasminCPD}), and thus the proof is complete.     
\end{proof}
\begin{remark}
When ${\bf y}$ is real-valued (AWGN channel) the only difference is the scaling of the second term, which should be $\frac{1}{c}$ instead of $\frac{1}{e}$, with $c:= (1+\frac{1}{N})^{||{\bf z}||^2_2}$ and 
\[
{\bf z}(n) := \left\{ \begin{array}{ll}
{\bf y}(n), & {\bf y}(n) > 0.5,\\
{\bf y}(n)-1, & {\bf y}(n) \leq 0.5.
\end{array}
\right.
\]
\end{remark}

\subsection{Codes over higher-order Galois Fields} Our framework can also handle the decoding of codes over higher-order Galois fields. For example, let $L=2^\ell$ and consider a system of parity equations over $\text{GF}(L)$. Such systems are of form ${\bf C} {\bf x} = {\bf q}$, where both ${\bf C}$ and ${\bf q}$ are given, and the equality is modulo $L=2^\ell$, i.e., ${\bf C} {\bf x} - {\bf q}$ is a vector of integer multiples of $L$. We can handle this type of equation by bringing in a familiar signal processing tool, namely, the complex roots of unity. That is, we seek to minimize over ${\bf x} \in \left\{0,\cdots,L-1\right\}^N$ the following cost function
\begin{equation}
\label{eq:HOGFminCPD}
\hspace{-1mm}-\sum_{m=1}^M \hspace*{-0.2mm}\text{Re}\left\{ e^{j \frac{2 \pi}{L} \left({\bf C}(m,:) {\bf x} - {\bf q}(m) \right)} \right\}\hspace{-0.5mm} + \frac{1}{\tilde c} \left( \hspace{-1mm}1 + \frac{1}{N}\right)^{\sum_{n=1}^N ({\bf y}(n)-{\bf x}(n))^2}\hspace{-10mm},
\end{equation}
where $j:=\sqrt{-1}$. 
The same logic applies for appropriately choosing $\tilde c$. 

\subsection{Least inconsistent solution of overdetermined linear system of equations over GF(2)}

So far, we have been dealing with underdetermined linear equations over a GF, and looking at minimum distance solutions relative to an ``anchor''. This is similar to the minimum norm solution of underdetermined linear equations in the real or complex field. The overdetermined version of the problem is also of interest, and has many applications -- e.g., in cryptanalysis \cite{Hastad1,Hastad2,HastadSlides}. It turns out that our framework can deal with this problem as well. We are given a set of $M>N$ linear equations in $N$ variables ${\bf G} {\bf x} = {\bf y}$ $\Leftrightarrow$ ${\bf G} {\bf x} + {\bf y} = {\bf 0}$ over $\text{GF}(2)$. The system is usually inconsistent, so we seek a ${\bf x}$ that minimizes the number of violated constraints, i.e., a solution that is least inconsistent. This can be simply posed as
\begin{equation}
\label{eq:OVERDETminCPD}
\min_{{\bf x} \in \left\{0,1\right\}^N} - \sum_{m=1}^M (-1)^{{\bf G}(m,:) {\bf x} + {\bf y}(m)}. 
\end{equation}

\subsection{Underdetermined or overdetermined?}
Every linear code can be generated as a linear combination of the columns of a code generating matrix ${\bf G}$. When  ${\bf G}$ is tall ($M \times N$ with $M > N$) and full column rank, the valid code words live in a subspace of dimension $N$ and they are orthogonal to the rows of a parity check matrix ${\bf C}$ of size $(M-N) \times M$. Given a noisy code word, we may pose the optimal decoding problem in two equivalent ways:
\begin{itemize}
    \item One is in code space, i.e., find a valid code word that is closest to the given noisy code word, and this is the formulation in (\ref{eq:MLPCCasminCPD}), wherein the unknown vector ${\bf x}$ is the sought clean code word. After solving (\ref{eq:MLPCCasminCPD}), we need to solve a consistent system of linear equations over GF(2) (via Gaussian elimination) to obtain the latent information sequence -- unless the code is systematic / in standard form wherein the information sequence appears as the prefix of the code sequence.     
    \item The other option is to pose the problem in the lower-dimensional space of the information sequence, i.e., find an information sequence that produces a valid code word that is closest to the given noisy code word; this is the formulation in (\ref{eq:OVERDETminCPD}), wherein the unknown vector ${\bf x}$ is the sought information sequence\footnote{In the problem statements, ${\bf x}$ is always taken to be an $N \times 1$ vector for consistency, but what is $N$ (length of code word or length of information sequence) changes depending on the context.}. When we take this route, there is no need for the additional Gaussian elimination step at the end, as we directly recover the information sequence. Note that the number of optimization variables is smaller and the rank of the CPD model is higher this way, but the complexity of our algorithms is linear in the number of unknowns and the rank, so this does not really affect the complexity of our approach. 
    \item The difference between the two ways of approaching the problem   lies in the initialization. If we go via (\ref{eq:MLPCCasminCPD}) there is a natural initialization for ${\bf x}$ -- the received noisy code word. Notice that this works for any parity check matrix. If we choose to solve (\ref{eq:OVERDETminCPD}) on the other hand, there is no obvious way to initialize the information sequence ${\bf x}$, unless the code is systematic -- in which case we read out a noisy version of ${\bf x}$ from the noisy code word itself. It is therefore preferable to use (\ref{eq:MLPCCasminCPD}) unless the code is systematic. The code is not systematic in cryptography applications for example.  
\end{itemize}

\subsection{Reprise} 

As we conclude this section, it is useful to reflect on what we learned. The take-home point is that there are many important problems which can be posed as instances of low-rank tensor minimization, for which it is not even necessary to perform tensor factorization -- the low-rank factors can be readily derived analytically, in closed-form. These clean-cut mappings of classic hard problems to instances of low-rank tensor minimization reinforce our hopes that, even in cases where the problem is more complicated and the cost function is not fully known (i.e., only examples / samples of the cost function are given), it may be possible to model those ``blind'' optimization problems using low-rank tensor factorization and minimization. 

Another important observation which stems from uniqueness of tensor completion, is that under certain conditions it is not even necessary to completely specify an instance of ILS or ILP  in the traditional sense of providing its input parameters $\textbf{H}, \textbf{b}$, $\textbf{c}$ in order to solve it. It is enough to specify the cost at certain (randomly chosen or systematic) points, and let tensor completion fill out the ``rest of the problem''. This possibility is certainly intriguing, and the direct result of uniqueness of low-rank tensor completion and our problem reformulation.   

\section{Experiments}

\subsection{Partition} 

Since we used the partition problem to establish the hardness of our problem in the worst case, let us compare one of the proposed continuous optimization algorithms for minCPD to an established approximation algorithm for the partition problem. For this purpose, we will use {\em greedy partitioning algorithm} \cite{MultiWayPartitioning} which first sorts the given numbers and then parses the sorted list from largest to smallest, assigning each to the bucket with the smallest running sum. This algorithm comes with a 7/6 approximation guarantee in terms of the  larger sum it outputs divided by the larger sum of an optimal partition. For smaller $N$ we also use enumeration to compute the optimal partition as another baseline.

The results obtained using the greedy algorithm and minCPD via Frank-Wolfe with adaptive stepsize (Algorithm \ref{algo:FW}) are summarized in Fig. \ref{fig:PartitionN20} for $N=20$ (with enumeration as another baseline) and Fig. \ref{fig:PartitionN30} for $N=30$ (without enumeration), for 100 Monte-Carlo trials each. 
In each trial, $N$ random integers in $\left\{1,\cdots,100\right\}$ are first drawn and then normalized to sum to $1$. For Frank-Wolfe, we used $C=5$, a hard-stop at a maximum of $1000$ iterations, and $5$ random initializations. 

It is clear that Algorithm \ref{algo:FW} outperforms the well-established greedy algorithm, and in many cases attains the optimal solution (zero subset imbalance, or the optimal subset imbalance obtained via enumeration, see Figs. \ref{fig:PartitionN20} and \ref{fig:PartitionN30}). 

\begin{figure}
	\centering
	\includegraphics[width=3in,keepaspectratio]{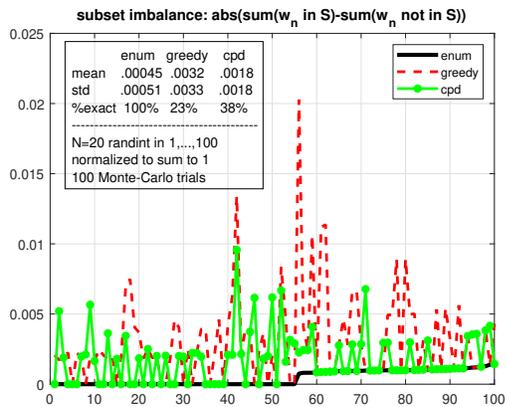}
	\caption{Partition gap for 100 problem instances with $N=20$ random integers in $\left\{1,\cdots,100\right\}$, normalized to sum to $1$. Greedy, optimal (enumeration-based), and proposed minimum CPD algorithm using Frank-Wolfe. Each point on the $x$ axis corresponds to a problem instance, and the instances are sorted in order of increasing partition gap for the optimal enumeration-based solution.}
	\label{fig:PartitionN20}
\end{figure}

\begin{figure}[t]
	\centering
	\includegraphics[width=3in,keepaspectratio]{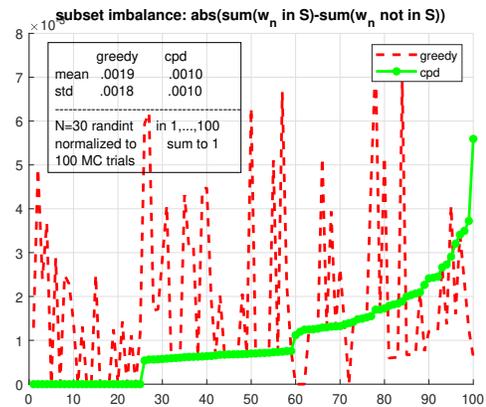}
	\caption{Partition gap for 100 problem instances with $N=30$ random integers in $\left\{1,\cdots,100\right\}$, normalized to sum to $1$. Greedy and proposed minimum CPD algorithm using Frank-Wolfe. Enumeration is too costly at $N=30$ ($2^{30}$ possible subsets), so the instances are sorted in order of increasing CPD partition gap.}
	\label{fig:PartitionN30}
\end{figure}

\subsection{Sign retrieval}

Our second set of experiments considers the application of our framework to the problem of sign retrieval. Towards this end, we use Algorithm \ref{algo:DGP} with stepsize parameter fixed to $0.1$ and a hard limit of $10^3$ gradient iterations. For initialization, we use the rank-one DP algorithm of Proposition \ref{prop:one}, which is used to efficiently determine the minimum of each rank-one factor. The best of these rank-one factor minima (the one that minimizes the full-rank cost) is then used to initialize Algorithm \ref{algo:DGP}. Ten additional random initializations are also used in case the DP initialization is not good enough. As a baseline, we use enumeration over all possible vectors of sign variables. 

For each setting of the sign retrieval problem parameters ($N=\text{length}({\bf x})$, $M=\text{length}({\bf y})$, $\sigma=\text{std}({\bf v}(m))$), we conduct 100 Monte-Carlo trials. For each trial, we draw random i.i.d. standard Gaussian ${\bf x}$ and ${\bf A}$, and i.i.d. zero-mean Gaussian noise ${\bf v}$ of standard deviation $\sigma_v$. 

For our first experiment, we choose $N=6$, $M=12$, and $\sigma_v=0.5$. Note that for this application the order of the tensor used in the CPD model is $M$ and its rank is $\frac{M(M-1)}{2}$. Fig. \ref{fig:SRN12M6sigma05cost} shows the cost function value attained for $100$ randomly drawn problem instances, constructed as specified above. Note that CPD comes very close to the optimal cost of enumeration, with a few minor spikes, for all instances considered. For this application, what is perhaps more important than the value of the cost function is the squared error between the ground-truth ${\bf x}$ and the $\hat {\bf x}$ estimated by a given algorithm. The values of squared error attained by CPD and enumeration are shown in Fig. \ref{fig:SRN12M6sigma05squarederror}. Notice that enumeration is by definition optimal in terms of the cost function, but not necessarily optimal in terms of instantaneous or even mean squared error (MSE) -- as the cost function is only a surrogate for MSE. Indeed, there are a few instances where CPD is better than enumeration in terms of squared error, and vice-versa. Overall though, the two approaches are very close in this experiment.  

\begin{figure}
	\centering	\includegraphics[width=3in,keepaspectratio]{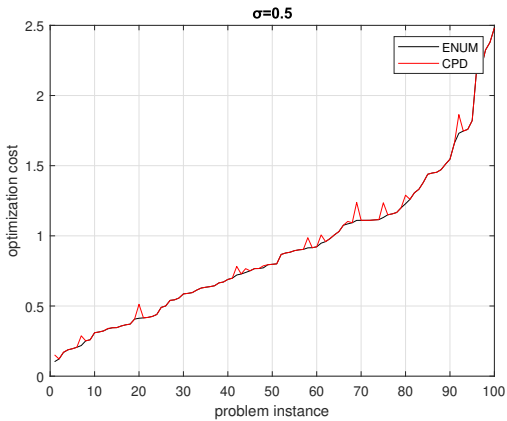}
	\caption{Sign retrieval: Cost attained by the enumeration-based solution and the proposed CPD-based approach, for $M=12$, $N=6$, $\sigma=0.5$. See text for details. Each point on the $x$ axis corresponds to a problem instance, and the instances are sorted in order of increasing cost for the optimal enumeration-based solution.}
	\label{fig:SRN12M6sigma05cost}
\end{figure}

\begin{figure}
	\centering	\includegraphics[width=3in,keepaspectratio]{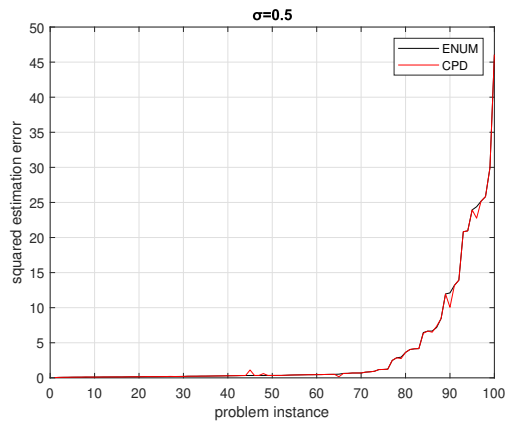}
	\caption{Sign retrieval: Squared ${\bf x}$-estimation error  attained by the enumeration-based solution and the proposed CPD-based approach, for $M=12$, $N=6$, $\sigma=0.5$. Each point on the $x$ axis corresponds to a problem instance, and the instances are sorted in order of increasing squared error of the enumeration-based solution.}
	\label{fig:SRN12M6sigma05squarederror}
\end{figure}

Monte-Carlo averages of the optimization cost and the squared error are depicted in Figs. \ref{fig:SRN12M6_cost_mse_function_sigma} and \ref{fig:SRN=2Msigma05_cost_mse_function_M} as a function of $\sigma$ and $N$ (with $M=2N$), respectively. We again observe the excellent performance of the proposed CPD approach in this set of experiments. 

\begin{figure}
	\centering	\includegraphics[width=3in,keepaspectratio]{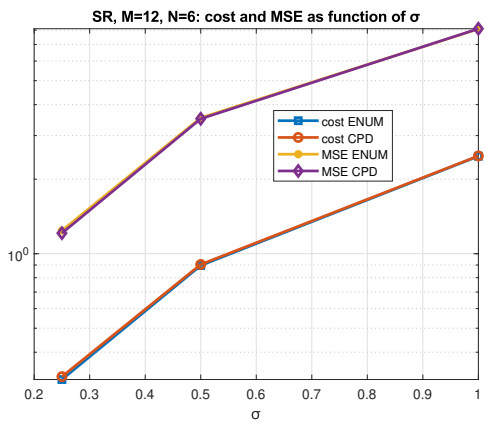}
	\caption{Sign retrieval: Mean optimization cost and MSE attained by the enumeration-based solution and the proposed CPD-based approach as a function of $\sigma$, for $M=12$, $N=6$.}
	\label{fig:SRN12M6_cost_mse_function_sigma}
 \end{figure}
\begin{figure}
	\centering	\includegraphics[width=3in,keepaspectratio]{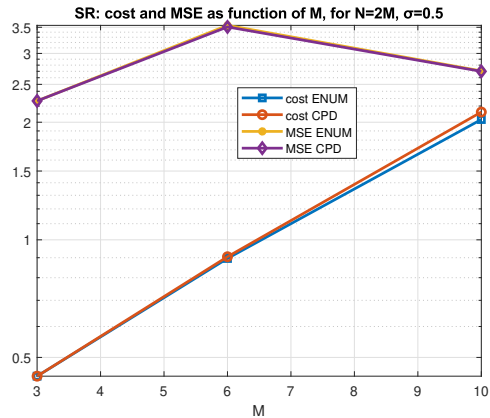}
	\caption{Sign retrieval: Mean optimization cost and MSE attained by the enumeration-based solution and the proposed CPD-based approach as a function of $N$, for $M=2N$ and $\sigma=0.5$.}
	\label{fig:SRN=2Msigma05_cost_mse_function_M}
\end{figure}

\subsection{Decoding parity check codes}

In our last set of experiments, we consider decoding five different rate-$\frac{1}{2}$ parity check codes, for different code densities and lengths $N$. For the first three scenarios $N=32$, while for the last two $N=96$. In the first scenario we use a custom LDPC code design\footnote{https://rptu.de/en/channel-codes/channel-codes-database/more-ldpc-codes} \cite{abu2010trapping}. In the remaining four scenarios, we use a systematic parity check matrix whose non-identity block is randomly generated from an i.i.d. $\left\{0,1\right\}$-Bernoulli distribution of density either $0.2$ or $0.8$. Each code is used for encoding i.i.d.  sequences of $N/2$ information bits, and the coded sequences are transmitted over BSCs with cross-over probability $p\in\left\{10^{-2.5}, 10^{-2}, 10^{-1.5}, 10^{-1}, 10^{-0.5}\right\}$. For each scenario and cross-over probability, we conduct $10,000$ Monte-Carlo runs and compare the decoding performance of our method to two baselines: i) enumeration (applicable only for $N=32$ due to its exponential complexity) and ii) belief propagation (BP) based decoding. Regarding our method, we use Algorithm 4 with step size equal to 0.05. The initialization of Algorithm 4 in terms of $\left\{\sigma_n\right\}_{n=1}^N$ and $\left\{\mu_n\right\}_{n=1}^N$ is $\sigma_n=0.5$, $\forall n$, while $\mu_n$ is set to the received noisy (possibly flipped) value of the corresponding code bit. Algorithm 4 is terminated when the number of iterations exceeds the limit of $2,000$ iterations or when the relative change of the objective drops below $10^{-9}$. As for the BP baseline, we use the \textit{ldpcDecode} function of MATLAB, to which the log-likelihood ratios based on the corresponding cross-over probabilities are provided as initialization, while the upper limit of iterations is set to $100$ (we did not observe any improvement beyond that). 

In Figs. (\ref{fig:PCdec1})-(\ref{fig:PCdec5}), we report the average Bit Error Rate (BER) for all the methods. We can observe that, in general, Algorithm 4 achieves better or comparable performance than BP. At low BSC cross-over probabibilities, we can see that Algorithm 4 outperforms BP for the custom LDPC code design of \cite{abu2010trapping}, and the high density parity check codes. In the rest of the cases the two methods attain comparable performance. BP has a small advantage for the longer low-density code in Fig. (\ref{fig:PCdec4}), as expected; BP works best with low-density long codes. We note that, unlike BP,  Algorithm 4 does not (need to) use the BSC cross-over probability, which is non-trivial to estimate for time-varying channels. We also note that BP exhibits a degradation in performance for the longer and denser code at low BSC error rates, see Fig. (\ref{fig:PCdec5}). This is repeatable (not an artifact of insufficient Monte-Carlo averaging) and likely due to the existence of many short loops in this case. On the other hand, BP is significantly faster than Algorithm 4. Overall though, given that Algorithm 4 is a completely new and application-agnostic take on a well-studied problem, the fact that it outperforms in terms of BER a proven MATLAB implementation of a custom-designed and widely used algorithm is satisfying. Matlab programs that can be used to reproduce the results in this subsection can be found in the companion supplementary material in IEEExplore.

\begin{figure}
	\centering	\includegraphics[width=3in,keepaspectratio]{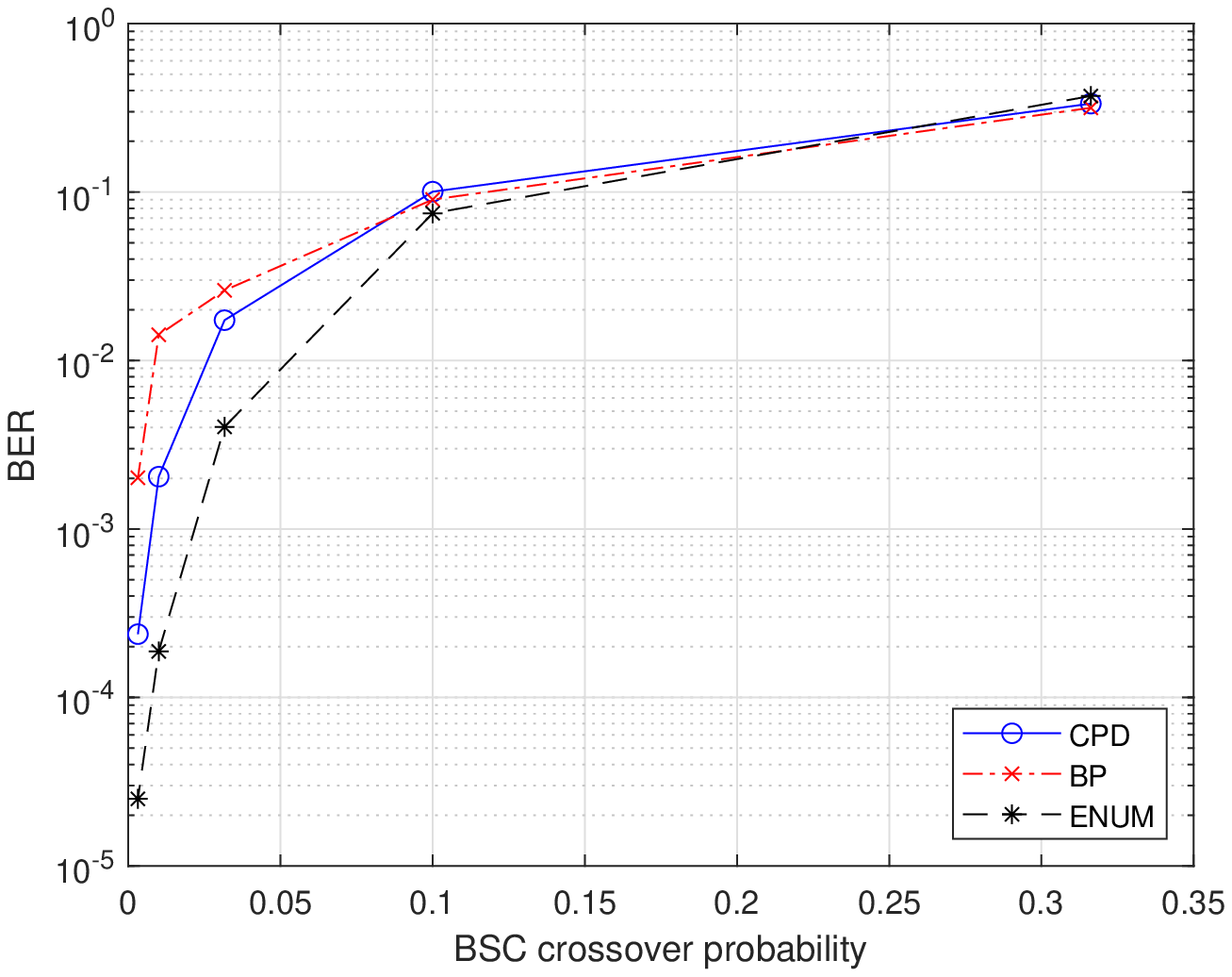}
	\caption{Parity check decoding: BER vs BSC cross-over probability for the first scenario: LDPC code of $N=32$ and $M=16$.}
	\label{fig:PCdec1}
 \centering	\includegraphics[width=3in,keepaspectratio]{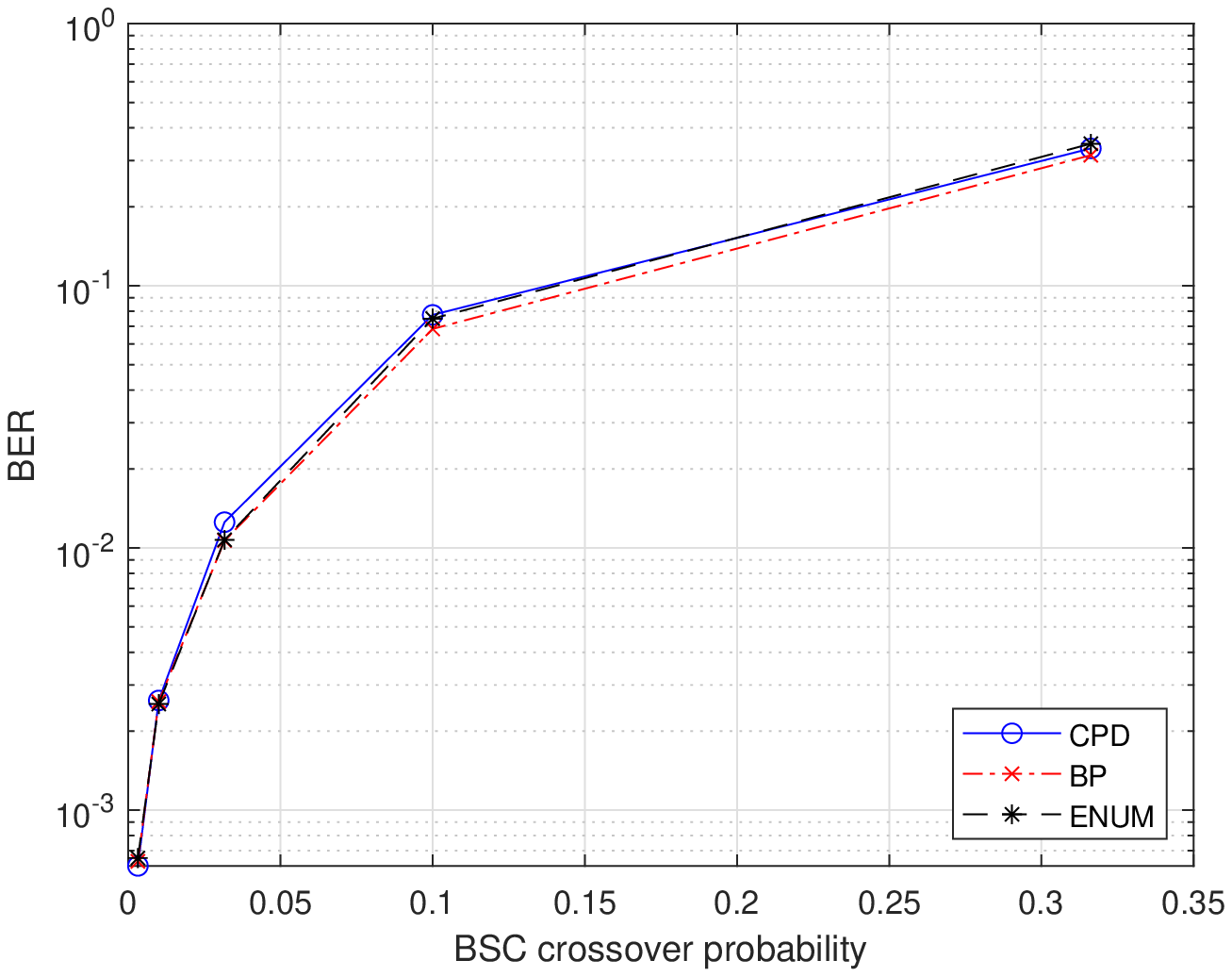}
	\caption{Parity check decoding: BER vs BSC cross-over probability for the second scenario: random code of $N=32$, $M=16$, and $20\%$ density.}
	\label{fig:PCdec2}
\end{figure}
\begin{figure}[t]
	\centering	\includegraphics[width=3in,keepaspectratio]{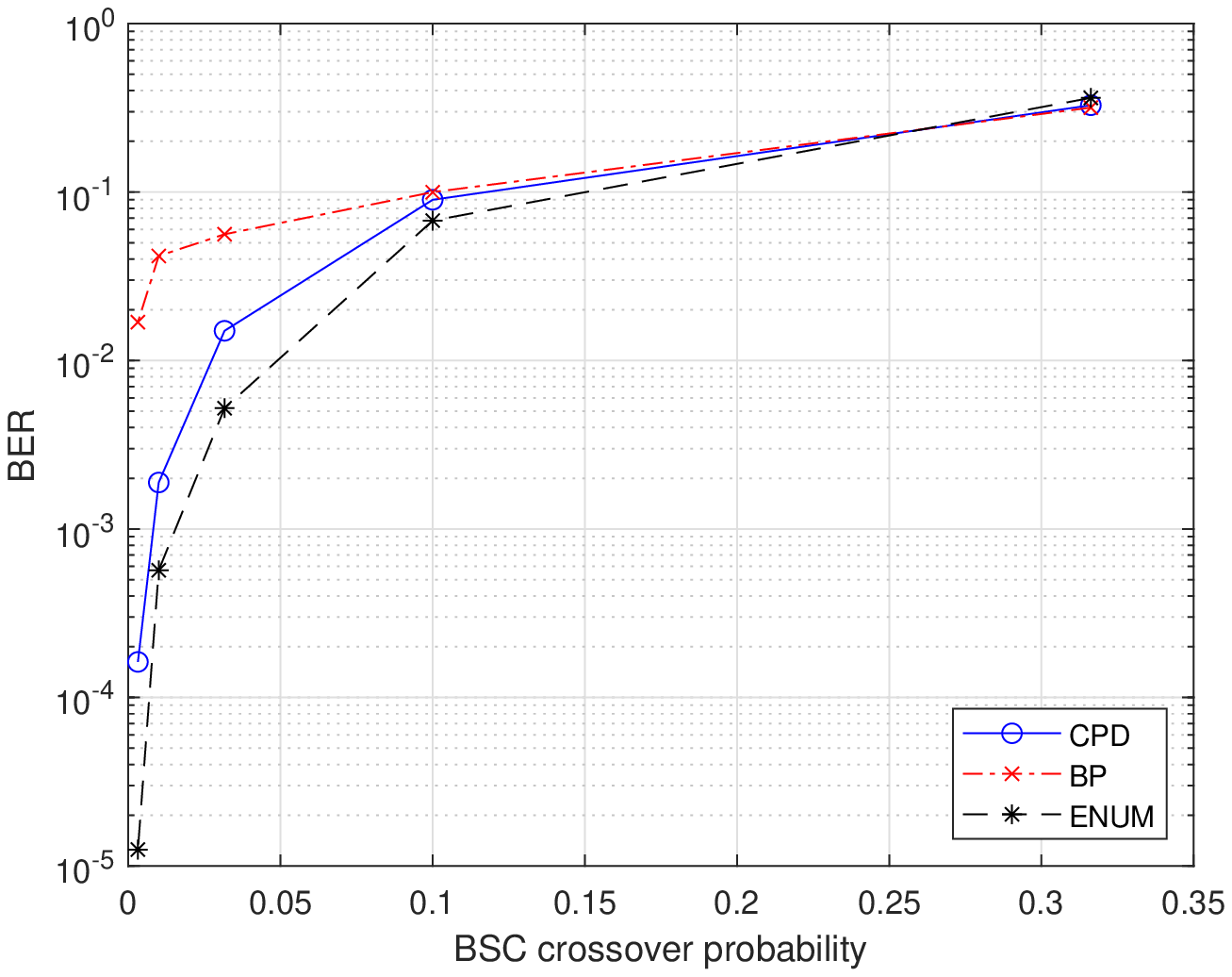}
	\caption{Parity check decoding: BER vs BSC cross-over probability for the  third scenario: random code with $N=32$, $M=16$ and $80\%$ density.}
	\label{fig:PCdec3}
	\centering	\includegraphics[width=3in,keepaspectratio]{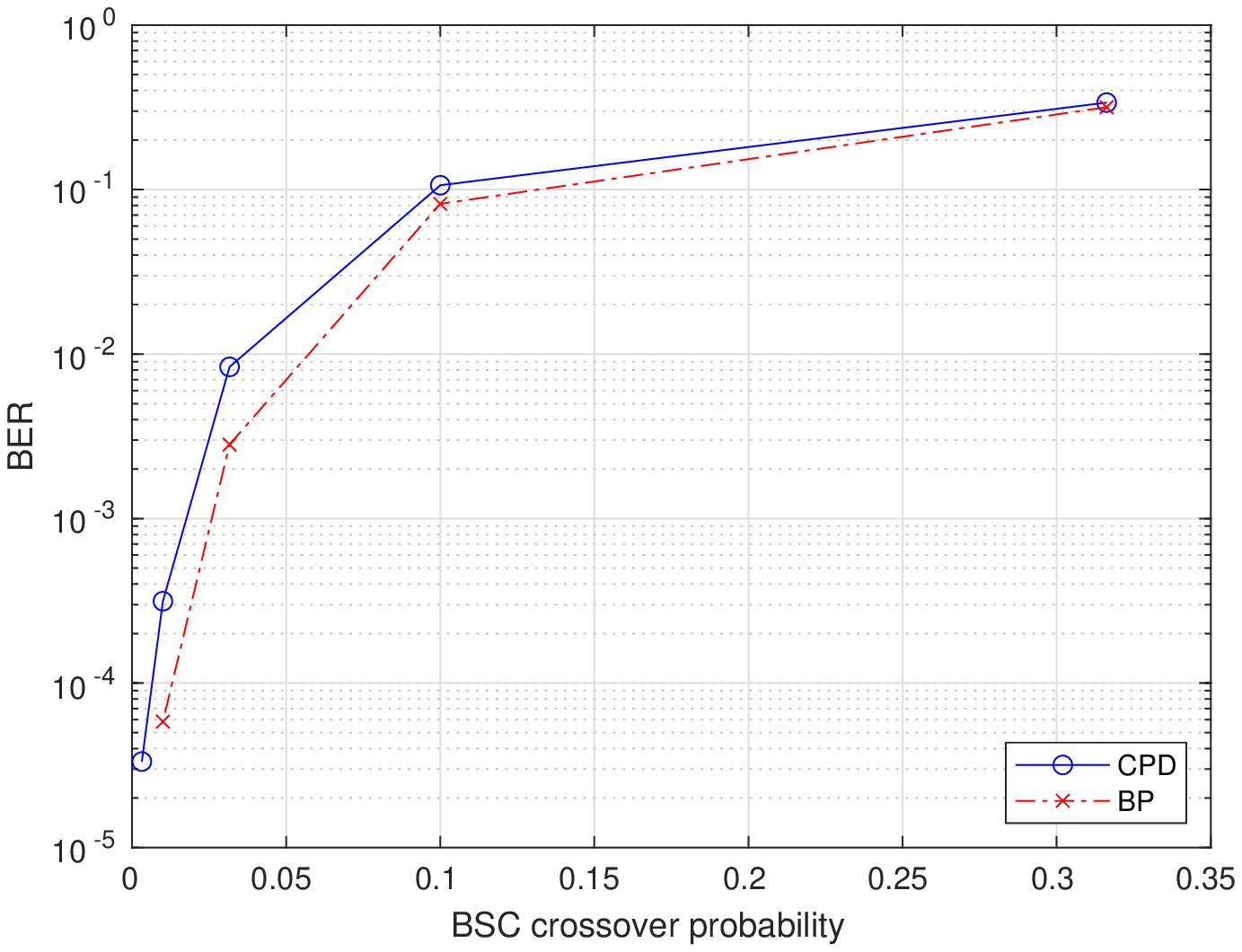}
	\caption{Parity check decoding: BER vs BSC cross-over probability for the fourth scenario: random code with $N=96$, $M=48$ and $20\%$ density.}
	\label{fig:PCdec4}
 \end{figure}
 \begin{figure}
	\centering	\includegraphics[width=3in,keepaspectratio]{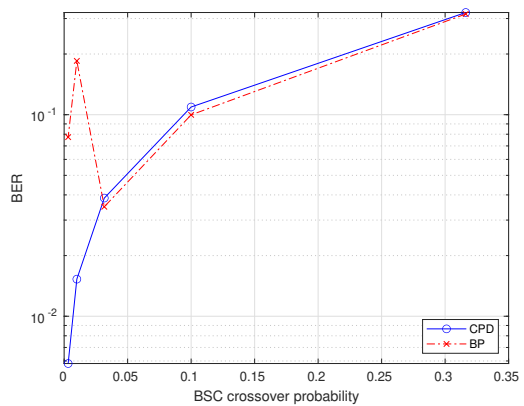}
	\caption{Parity check decoding: BER vs BSC cross-over probability for the fifth scenario: random code with $N=96$, $M=48$ and $80\%$ density.}
	\label{fig:PCdec5}
\end{figure}

\section{Conclusions}

We have considered a fundamental tensor problem and showed that it is NP-hard. While most tensor problems are NP-hard \cite{10.1145/2512329}, it is surprising to see that our particular problem is NP-hard for rank as small as two, but not for rank equal to one. We note here that the best rank-one least squares approximation problem for tensors is already NP-hard.

Given the discrete / combinatorial nature of the problem considered, it is also unexpected to see that it admits an equivalent continuous reformulation. While this reformulation is itself NP-hard by virtue of equivalence, it opens the door for gradient-based approaches from the nonconvex continuous optimization literature. 

We have shown that an impressive variety of hard optimization problems that are widely used in engineering can be posed as special instances of our problem of interest. These include integer least squares, integer linear and quadratic programming, certain mixed integer programming problems, and solving systems of underdetermined and overdetermined linear equations over Galois fields. For all these problems, the low-rank factorization needed to set up the optimization problem is available analytically, in simple closed-form. There is no need for tensor factorization. 

As tangible signal processing and communications engineering applications, we delved into sign retrieval and the decoding of linear parity check codes. We have shown that the performance of the proposed suite of gradient-based approaches is surprisingly good in many of these applications, and under certain conditions it can beat tried-and-proven application-specific algorithms that come with certain performance guarantees. This success is sometimes dependent on using a suitable initialization, either application-specific (``dirty'' code bits in the case of parity decoding) or ``universal'' (the DP algorithm used to initialize the gradient iterations for sign retrieval) in other cases. For other problems, like the partition, a few random initializations seem to work well.  

Our main results (hardness, equivalence of continuous reformulation) are also applicable to all popular tensor models beyond CPD, including Tucker/HOSVD, TT, and TR. 

Note that for the families of problems and applications considered in this paper, $R$ is a small constant or linear / very low-order polynomial function of $N$. For so-called {\em black box} optimization problems out in the wild, the worst-case $R$ can be exponential in $N$, and we would have to use low-rank approximation to keep complexity at bay. This is not an issue though for the various problems considered in this paper. 

We are currently working on further improving scalability and speed, establishing convergence for some of the algorithmic variants, and considering applications in a variety of settings. Performance analysis is naturally of interest, but is likely to be application-specific. We are also considering top-$k$ / bottom-$k$ extensions which are appropriate for top-$k$ recommendation and other applications. We hope to report on these directions in forthcoming work. 

\section{Acknowledgements}
Thanks to Prof. Farzad Farnoud-Hassanzadeh, who provided useful insights on linear GF coding and decoding, and to the reviewers for appreciating this work and offering useful suggestions.  

\bibliographystyle{IEEEtran}
\bibliography{IEEEabrv,refrences}

\begin{IEEEbiography}[{\includegraphics[scale=0.15]{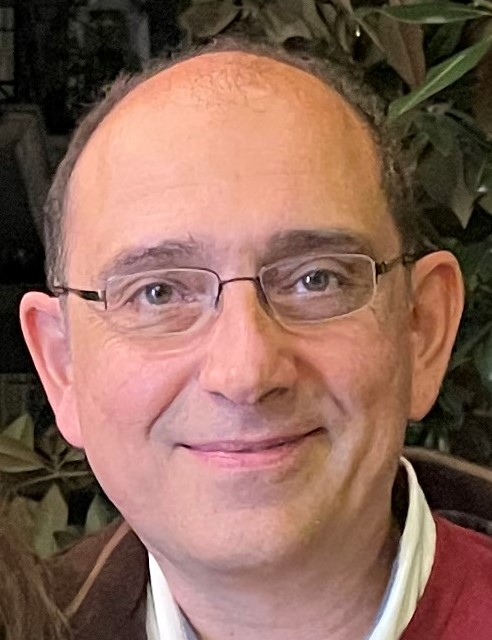}}]
{Nicholas D. Sidiropoulos} (Fellow, IEEE) received the Diploma in electrical engineering from the Aristotle University of Thessaloniki, Thessaloniki, Greece, and the M.S. and Ph.D. degrees in electrical engineering from the University of Maryland at College Park, College Park, MD, USA, in 1988, 1990, and 1992, respectively. He is the Louis T. Rader Professor in the Department of ECE at the University of Virginia. He has previously served on the faculty at the University of Minnesota and the Technical University of Crete, Greece. His research interests are in signal processing, communications, optimization, tensor decomposition, and machine learning. He received the NSF/CAREER award in 1998, the IEEE Signal Processing Society (SPS) Best Paper Award in 2001, 2007, 2011, and 2023, and the IEEE SPS Donald G. Fink Overview Paper Award in 2023. He served as IEEE SPS Distinguished Lecturer (2008–2009), Vice President - Membership (2017–2019) and Chair of the SPS Fellow evaluation committee (2020-2021). He received the 2010 IEEE SPS Meritorious Service Award, the 2013 Distinguished Alumni Award of the ECE Department at the University of Maryland, the 2022 EURASIP Technical Achievement Award, and the 2022 IEEE SPS Claude Shannon - Harry Nyquist Technical Achievement Award. He is a fellow of EURASIP (2014).
\end{IEEEbiography}

\begin{IEEEbiography}[{\includegraphics[scale=0.2, angle=-90, clip, trim = {14cm 9.5cm 12cm 10cm}]{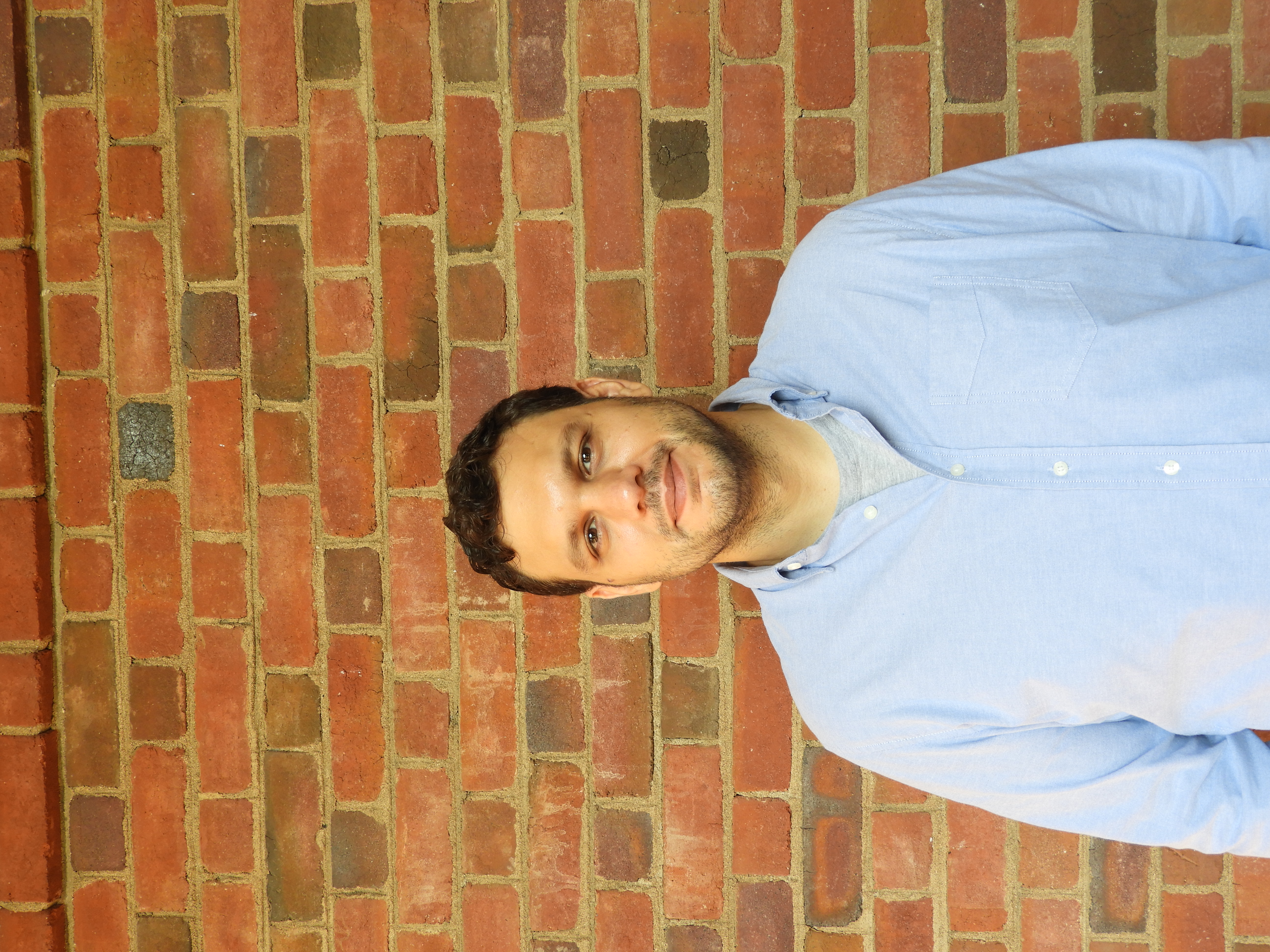}}]{Paris A. Karakasis}received the Diploma and M.Sc. degree in Electrical and Computer Engineering from the Technical University of Crete, Chania, Greece, in 2017 and 2019, respectively. Currently, he is a Ph.D. student at the Electrical and Computer Engineering Department, University of Virginia, Charlottesville, VA, USA. His research interests include signal processing, optimization, machine learning, tensor decomposition, and graph mining.
\end{IEEEbiography}

\begin{IEEEbiography}[{\includegraphics[scale=0.065]{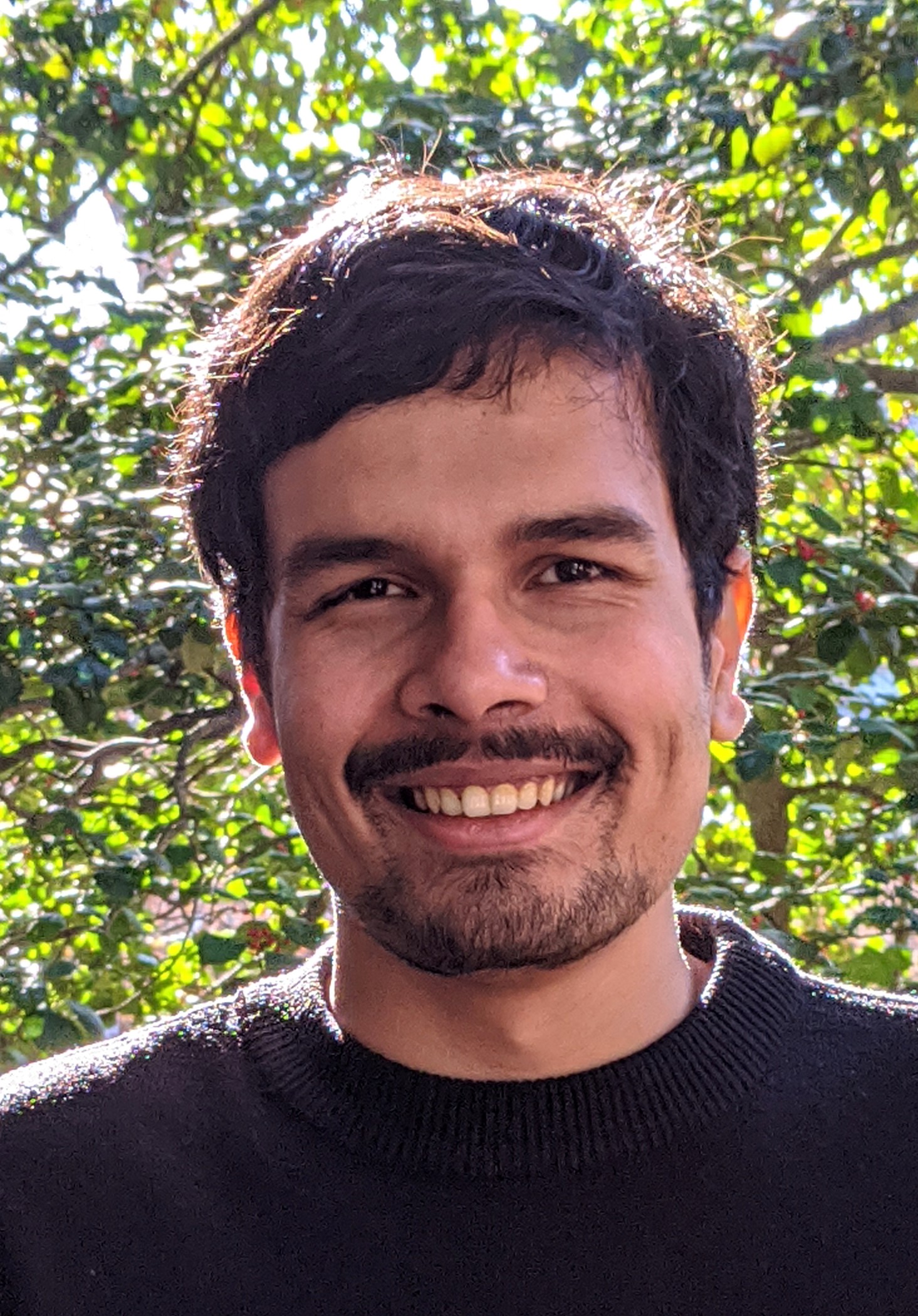}}]
{Aritra Konar} is an Assistant Professor in the Department of Electrical Engineering at KU Leuven, Leuven, Belgium. He received the B.Tech. degree in electronics and communications engineering from the West Bengal University of Technology, West Bengal, India, and the M.S. and Ph.D. degrees in electrical engineering from the University of Minnesota, Minneapolis, MN, USA, in 2011, 2014, and 2017 respectively. From 2017-2022, he was a postdoctoral researcher with the Department of ECE, University of Virginia, Charlottesville, VA, USA. His research interests include signal processing, graph mining, nonlinear optimization and data analytics.
\end{IEEEbiography}

\end{document}